\documentclass{article}

\newcommand{\Exp}{\mathop{\mathbb{E}}}

\newcommand{\cD}{\mathcal{D}}
\newcommand{\cE}{\mathcal{E}}
\newcommand{\cF}{\mathcal{F}}

\newcommand{\cY}{\mathcal{Y}}
\newcommand{\cN}{\mathcal{N}}
\newcommand{\N}{\mathbb{N}}

\newcommand{\BZ}{\mathbb{Z}}

\newcommand{\unif}{\mathsf{Unif}}

\newcommand{\mcsp}{\textsf{Max-CSP}}
\newcommand{\mocsp}{\textsf{Max-OCSP}}
\newcommand{\mfas}{\textsf{MFAS}}
\newcommand{\ord}{\textsf{ord}}

\newcommand{\sym}{\mathfrak{S}}
\renewcommand{\unif}{\mathsf{Unif}}

\newcommand{\mas}{\textsf{MAS}}
\newcommand{\mug}{\textsf{Max-UniqueGames}}
\newcommand{\mbtwn}{\textsf{Max-Btwn}}
\newcommand{\btwn}{\textsf{Btwn}}

\newcommand{\cspval}{\textsf{csp-val}}
\newcommand{\ocspval}{\textsf{ocsp-val}}

\newcommand{\Alg}{\mathsf{Alg}}
\newcommand{\Prot}{\mathsf{Prot}}
\newcommand{\Player}{\mathsf{Player}}
\newcommand{\supp}{\textsf{supp}}
\newcommand{\yes}{\textbf{YES}}
\newcommand{\no}{\textbf{NO}}

\newcommand{\IRMD}{\mathsf{IRMD}}

\newcommand{\veca}{\mathbf{a}}
\newcommand{\vecb}{\mathbf{b}}

\newcommand{\vecj}{\mathbf{j}}

\newcommand{\vecv}{\mathbf{v}}

\newcommand{\vecy}{\mathbf{y}}
\newcommand{\vecz}{\mathbf{z}}

\newcommand{\vecsigma}{{\boldsymbol{\sigma}}}

\newcommand{\vecpi}{\boldsymbol{\pi}}

\renewcommand{\tilde}{\widetilde}

\newcommand{\bigbar}{\;\middle|\;}

\newcommand{\matchings}{\mathcal{M}}

\newcommand\eqdef{\stackrel{\text{\small def}}{=}}

\newcommand{\mF}{\mocsp(\cF)}
\newcommand{\coar}{\downarrow}
\newcommand{\qcoar}{\coar q}
\newcommand{\refn}{\uparrow}
\newcommand{\DY}{\cY^{\Pi,\vecpi}_{q,\alpha,T}(n)}
\newcommand{\DN}{\cN^{\Pi}_{q,\alpha,T}(n)}
\newcommand{\Piq}{\Pi^{\qcoar}}
\newcommand{\Psiq}{\Psi^{\qcoar}}
\newcommand{\mPi}{\mocsp(\Pi)}
\newcommand{\mPiq}{\mcsp(\Piq)}

\newcommand{\NP}{\mathbf{NP}}

\usepackage[margin=1in]{geometry}
\usepackage[utf8]{inputenc}
\usepackage{bm}
\usepackage{graphicx}
\usepackage{color,xcolor}
\usepackage{amsmath,amsfonts, amssymb, amsthm, thmtools}
\usepackage{algorithm, algpseudocode}

\PassOptionsToPackage{hyphens}{url}
\usepackage[colorlinks=true, allcolors=blue]{hyperref}
\usepackage[capitalise,nameinlink]{cleveref}

\usepackage[style=alphabetic, backend=biber, minalphanames=3, maxalphanames=4, maxbibnames=99, mincitenames=5, maxcitenames=5]{biblatex}

\crefformat{section}{#2\S#1#3}
\crefformat{subsection}{#2\S#1#3}
\crefformat{subsubsection}{#2\S#1#3}

\declaretheoremstyle[bodyfont=\it,qed=\qedsymbol]{noproofstyle}

\numberwithin{equation}{section}

\declaretheorem[name=Observation,numbered=no]{observation*}

\declaretheorem[numberlike=equation]{fact}

\declaretheorem[numberlike=equation]{theorem}

\declaretheorem[name=Theorem,numbered=no]{theorem*}

\declaretheorem[numberlike=equation]{lemma}
\declaretheorem[name=Lemma,numbered=no]{lemma*}

\declaretheorem[name=Corollary,numbered=no]{corollary*}

\declaretheorem[name=Proposition,numbered=no]{proposition*}

\declaretheorem[name=Claim,numbered=no]{claim*}

\declaretheorem[name=Conjecture,numbered=no]{conjecture*}

\declaretheorem[name=Question,numbered=no]{question*}

\declaretheoremstyle[bodyfont=\it]{defstyle} 

\declaretheorem[numberlike=equation,style=defstyle]{definition}
\declaretheorem[unnumbered,name=Definition,style=defstyle]{definition*}

\declaretheorem[unnumbered,name=Example,style=defstyle]{example*}

\declaretheorem[unnumbered,name=Notation=defstyle]{notation*}

\declaretheorem[unnumbered,name=Construction,style=defstyle]{construction*}

\declaretheoremstyle[]{rmkstyle} 

\newtheorem*{remark}{Remark}

\usepackage{tikz}
\usetikzlibrary{arrows.meta}
\usepackage[font=footnotesize]{caption}
\usepackage[font=footnotesize]{subcaption}

\title{Streaming approximation resistance of every ordering CSP}
\author{Noah G. Singer\thanks{Department of Computer Science, Carnegie Mellon University, Pittsburgh, PA, USA. Email: \texttt{ngsinger@cs.cmu.edu}.} 
\and Madhu Sudan\thanks{School of Engineering and Applied Sciences, Harvard University, Cambridge, Massachusetts, USA. Email: \texttt{madhu@cs.harvard.edu}.}
\and Santhoshini Velusamy\thanks{Toyota Technological Institute, Chicago, Illinois, USA. Email: \texttt{santhoshini@ttic.edu}.}}

\begin{document}

\maketitle


\begin{abstract}
    An ordering constraint satisfaction problem (OCSP) is defined by a family $\cF$ of predicates mapping permutations on $\{1,\ldots,k\}$ to $\{0,1\}$. An instance of $\mF$ on $n$ variables consists of a list of constraints, each consisting of a predicate from $\cF$ applied on $k$ distinct variables. The goal is to find an ordering of the $n$ variables that maximizes the number of constraints for which the induced ordering on the $k$ variables satisfies the predicate. OCSPs capture well-studied problems including `maximum acyclic subgraph' ($\mas$) and ``maximum betweenness''. 

    In this work, we consider the task of approximating the maximum number of satisfiable constraints in the (single-pass) streaming setting, when an instance is presented as a stream of constraints. We show that for every $\cF$, $\mF$ is approximation-resistant to $o(n)$-space streaming algorithms, i.e., algorithms using $o(n)$ space cannot distinguish streams where almost every constraint is satisfiable from streams where no ordering beats the random ordering by a noticeable amount. This space bound is tight up to polylogarithmic factors. In the case of $\mas$ our result shows that for every $\epsilon>0$, $\mas$ is not $(1/2+\epsilon)$-approximable in $o(n)$ space. The previous best inapproximability result, due to \textcite{GT19}, only ruled out $3/4$-approximations in $o(\sqrt n)$ space.
    
    Our results build on recent works of \textcite{CGSV24,CGS+22-linear-space} who provide a tight, linear-space inapproximability theorem for a broad class of ``standard'' (i.e., non-ordering) constraint satisfaction problems (CSPs) over arbitrary (finite) alphabets.
    Our results are obtained by building a family of appropriate standard CSPs (one for every alphabet size $q$) from any given OCSP and applying their theorem to this family of CSPs. To convert the resulting hardness results for standard CSPs back to our OCSP, we show that the hard instances from this earlier theorem have the following ``partition expansion'' property with high probability: For every partition of the $n$ variables into small blocks, for most of the constraints, all variables are in distinct blocks.
\end{abstract}

\section{Introduction}

In this work, we consider the complexity of ``approximating'' ``ordering constraint satisfaction problems (OCSPs)'' in the ``streaming model''. We introduce these notions below before describing our results.

\subsection{Orderings and constraint satisfaction problems}\label{ssec:intro-notation}

In this work, we consider optimization problems where the solution space is all possible orderings of $n$ variables. The Travelling Salesperson Problem and most forms of scheduling problems fit this description, though our work considers a more concrete class of problems, namely \emph{ordering constraint satisfaction problems (OCSPs)}. OCSPs as a class were first defined by \textcite{GHM+11}. To describe them here, we first set up some notation and terminology, and then give some examples.

We let $[n]$ denote the set $\{1,\ldots,n\}$ and $\sym_n$ denote the space of all permutations in $[n]^n$, i.e., \[ \sym_n \eqdef \{ \vecsigma = (\sigma_1,\ldots,\sigma_n ) \in [n]^n : \forall i \neq j, \sigma_i \neq \sigma_j \}. \] We interpret each element $\vecsigma \in \sym_n$ as a schedule for $n$ tasks, labeled $1, \ldots, n$, such that task $i$ is scheduled in position $\sigma_i$. We use bold type to denote vectors (e.g., $\vecsigma$), parenthetical indices for sequences of vectors (e.g., $\vecsigma(1),\ldots,\vecsigma(m)$), and normal type to denote scalar entries (e.g., $\sigma_i$).

Given $k$ distinct integers $a_1,\ldots,a_k$, we define $\ord(a_1,\ldots,a_k) \in \sym_k$ as the unique $\vecpi = (\pi_1,\ldots,\pi_k) \in\sym_k$ such that $\pi_i < \pi_j$ iff $a_i < a_j$ for all $i \neq j$. If $a_1,\ldots,a_k$ are not all distinct, we write $\ord(a_1,\ldots,a_k) = \bot$, and thus we can view $\ord$ as a map $\BZ^k \to \sym_k \cup \{\bot\}$. Given $\vecsigma = (\sigma_1,\ldots,\sigma_n) \in \sym_n$ and $k$ indices $\vecj = (j_1,\ldots,j_k) \in [n]^n$, we let $\vecsigma|_\vecj$ denote $(\sigma_{j_1},\ldots,\sigma_{j_k})\in[n]^n$.

The solution space of OCSPs is precisely $\sym_n$. A {\em $k$-ary ordering constraint predicate} is a function $\Pi:\sym_k \to \{0,1\}$. An {\em ordering constraint application} $(\Pi,\vecj)$ on $n$ variables is given by a predicate $\Pi$ and a $k$-tuple $\vecj \in [n]^n$ of distinct indices, and $(\Pi,\vecj)$ is {\em satisfied} by an assignment $\vecsigma\in \sym_n$ iff $\Pi(\ord(\vecsigma|_{\vecj}))=1$. In the interest of brevity, we will often skip the term ``ordering'' below and further refer to constraint predicates as ``predicates'' and constraint applications as ``constraints''.

A \emph{maximum ordering constraint satisfaction problem}, denoted $\mF$, is specified by a (finite) family of ordering constraint predicates  $\cF \subseteq \bigcup_{k \in \N} \{ \Pi : \sym_k \rightarrow \{0,1\}\}$. An {\em instance} of $\mF$ on $n$ variables is given by $m$ constraints $C_1,\ldots,C_m$ where $C_i = (\Pi_i,\vecj(i))$ and $\Pi_i \in \cF$. (We will typically specialize to the case where the family $\cF$ contains only a single predicate $\Pi$; in this case, we write the problem as $\mPi$ and omit $\Pi$ from constraint descriptions. We'll see below that for proving inapproximability results, it's sufficient to consider this case; see the Remark in \ref{sec:approximability} and the proof of \ref{thm:main} in \ref{sec:lower-bound}.) The \emph{value} of an ordering $\vecsigma \in \sym_n$ on the instance $\Psi$, denoted $\ocspval_\Psi(\vecsigma)$, is the fraction of constraints satisfied by $\vecsigma$, i.e., \[ \ocspval_\Psi(\vecsigma) \eqdef \tfrac{1}{m}\sum_{i\in[m]} \Pi_i(\ord(\vecsigma|_{\vecj(i)})). \] The optimal value of $\Psi$ is defined as \[ \ocspval_\Psi \eqdef \max_{\vecsigma \in \sym_n}\{\ocspval_\Psi(\vecsigma)\}. \]

The canonical problem that fits the $\mocsp$ framework is the \emph{maximum acyclic subgraph} ($\mas$) problem. In this problem, the input is a directed graph on $n$ vertices, and the goal is to find an ordering of the vertices that maximizes the number of forward edges. A simple depth-first search algorithm can decide whether a given graph $G$ has a \emph{perfect} ordering (i.e., one which has \emph{no} backward edges); however, \textcite{Kar72}, in his famous list of 21 $\NP$-complete problems, proved the $\NP$-completeness of deciding whether, given a graph $G$ and a parameter $k$, there exists an ordering of the vertices such that at least $k$ edges are forward. For our purposes, $\mas$ can be viewed as a 2-ary OCSP $\mas=\mocsp(\Pi_{\mas})$, where $\Pi_{\mas}: \sym_2 \rightarrow \{0,1\}$ denotes the predicate given by $\Pi_{\mas}(1,2) = 1$ and $\Pi_{\mas}(2,1) = 0$, and we associate vertices with variables and edges with constraints. Indeed, a constraint $(j_1,j_2)$ (where $j_1,j_2 \in [n]$ are distinct variables) will be satisfied by an ordering $\vecsigma =(\sigma_1,\ldots,\sigma_n)\in \sym_n$ iff $\Pi_{\mas}(\ord(\vecsigma|_{(j_1,j_2)})) = 1$, or equivalently, iff $\sigma_{j_1} < \sigma_{j_2}$. In the ``scheduling'' interpretation of OCSPs, a constraint $(j_1,j_2)$ expresses ``precedence'' of event $j_1$ over event $j_2$, since it is satisfied iff $j_1$ is scheduled before $j_2$.

A second natural $\mocsp$ problem is the \emph{maximum betweenness} ($\mbtwn$) problem. This is a 3-ary OCSP in which an ordering $\vecsigma = (\sigma_1,\ldots,\sigma_n)$ satisfies a constraint $(j_1,j_2,j_3)$ iff $\sigma_{j_2}$ is between $\sigma_{j_1}$ and $\sigma_{j_3}$, i.e., iff $\sigma_{j_1} < \sigma_{j_2} < \sigma_{j_3}$ \emph{or} $\sigma_{j_1} > \sigma_{j_2} > \sigma_{j_3}$, and the goal is again to find the maximum number of satisfiable constraints. Thus, $\mbtwn=\mocsp(\Pi_\btwn)$ where we define the constraint predicate $\Pi_{\btwn}:\sym_3 \rightarrow \{0,1\}$ by $\Pi_\btwn(1,2,3) = 1, \Pi_\btwn(3,2,1) = 1$, and $\Pi_\btwn(\vecpi) = 0$ for all other $\vecpi \in \sym_3$. The complexity $\mbtwn$ was originally studied by \textcite{Opa79}, who proved that even deciding whether a set of betweenness constraints is perfectly satisfiable (i.e., whether the value of an instance $\Psi$ is $1$) is $\NP$-complete.

\subsection{Approximability}\label{sec:approximability}

In this work, we consider the \emph{approximability} of ordering constraint satisfaction problems. We say that a (randomized) algorithm $\Alg$ is an {\em $\alpha$-approximation algorithm} for $\mF$ if for every instance $\Psi$, $\alpha \cdot \ocspval_\Psi \leq \Alg(\Psi) \leq \ocspval_\Psi$ with probability at least 2/3 over the internal coin tosses of $\Alg$. Thus our approximation factors $\alpha$ are numbers in the interval $[0,1]$. 

Given an ordering predicate $\Pi:\sym_k \to \{0,1\}$, let $\rho(\Pi) = \frac{|\{\vecpi \in \sym_k: \Pi(\vecpi) = 1\}|}{k!}$ denote the probability that $\Pi$ is satisfied by a random ordering. Given a (finite) family of predicates $\cF$, let $\rho(\cF) = \min_{\Pi \in \cF} \{\rho(\Pi)\}$. Every instance $\Psi$ of $\mF$ satisfies \[ \ocspval_\Psi \geq \rho(\cF) \] (since the right-hand side is a lower bound on the expected value of a random assignment). Thus, the trivial algorithm that always outputs $\rho(\cF)$ is a $\rho(\cF)$-approximation algorithm for $\mF$. Under what conditions it is possible to beat this ``trivial'' approximation is a major open question.

\begin{remark}
We define $\rho(\cF) = \min_{\Pi \in \cF} \{\rho(\Pi)\}$ to be the ``trivial'' approximability threshold for $\mF$ because for every $\epsilon > 0$ there are instances of $\mF$ with value at most $\rho(\cF) + \epsilon$. This is a consequence, for instance, of \ref{lemma:n-upper-bound} below, which holds \emph{a priori} for the single-predicate case $|\cF| = 1$, but can be extended to general finite families $\cF$ by taking the minimum over $\Pi \in \cF$ of $\rho(\Pi)$, since every instance of $\mPi$ is also an instance of $\mF$ with the same value.
\end{remark}

A problem is said to be {\em approximation resistant} with respect to a given class of algorithms if the trivial algorithm is essentially the best. Specifically for $\mocsp(\cF)$, we say it is approximation resistant for a class of algorithms if for every $\epsilon > 0$, no algorithm in the class $(\rho(\cF)+\epsilon)$-approximates $\mocsp(\cF)$. Ordering CSP problems were shown to be approximation resistant with respect to the class of polynomial time algorithms by \textcite{GHM+11} (assuming the unique games conjecture (UGC) of \textcite{Kho02}). In this work, we consider the analogous question with respect to ``sublinear-space streaming algorithms'', which we define next.

\subsection{Streaming algorithms}

A (single-pass) streaming algorithm for OCSPs is defined as follows. In $\mF$, an instance $\Psi$ is presented as a stream $(C_1,\ldots,C_m)$, where each stream element is a constraint $C_i =(\Pi_i,\vecj(i))$. A streaming algorithm $\Alg$ updates its state with each element of the stream and at the end produces an output $\Alg(\Psi) \in [0,1]$ (which is supposed to estimate $\ocspval_\Psi$). The measure of complexity of interest to us is the space used by $\Alg$ measured as a function of $n$, the number of variables in $\Psi$. (This is a somewhat non-standard choice in general but standard in the CSP literature. This choice is consistent with that of measuring the complexity of graph algorithms as a function of the number of vertices in the input instance.)
In particular, we distinguish between algorithms that use space polylogarithmic in $n$ and space that grows polynomially ($\Omega(n^\delta)$ for $\delta > 0$).
(Note that for this coarse level of distinction, measuring space as a function of $n$ or as a function of the input length would be qualitatively equivalent. However, our main result is more detailed and tight up to polylogarithmic factors when viewed as a function of $n$.)

We say that a problem $\mF$ is {\em (streaming) approximable} if we can beat the trivial $\rho(\cF)$-approximation algorithm by a positive constant factor. Specifically, $\mF$ is said to be approximable if for every $\delta > 0$ there exists $\epsilon > 0$ and a space $O(n^{\delta})$ algorithm that $(\rho(\cF)+\epsilon)$-approximates $\mF$. We say $\mPi$ is {\em (streaming) approximation-resistant} otherwise.

In recent years, investigations into CSP approximability in the streaming model have been strikingly successful, resulting in tight characterizations of streaming approximability for many problems \cite{KK15,KKS15,KKSV17,GVV17,GT19,KK19,CGV20,CGSV21-boolean,CGS+22-linear-space,CGS+22-monarchy,SSSV23-dicut,SSSV23-random-ordering,CGSV24}. Most of these papers study approximability, not of \emph{ordering} CSPs, but of ``standard'' CSPs where the variables can take values in a finite alphabet. (\cite{GVV17} and \cite{GT19} are the exceptions, and we will discuss them below.) Single-pass streaming algorithms with subpolynomial space are not formally a subclass of polynomial time algorithms.\footnote{The models are incomparable in the $\Omega(n^\delta)$ space regime since the streaming model has no time complexity or uniformity assumptions.} However, as far as we are aware, all known sublinear-space streaming algorithms for CSP approximation can be implemented as polynomial-time (often even linear-time!) algorithms. Indeed, there is essentially only one family of techniques for achieving nontrivial approximation ratios for CSPs via sublinear-space streaming algorithms, namely algorithms counting ``biases'' (see \cite{GVV17,CGV20,CGSV21-boolean,SSSV23-dicut}), and these algorithms can be implemented as linear-time classical algorithms and only give nontrivial guarantees for small classes of CSPs. In particular, \textcite{SSSV23-dicut} showed that the $\textsf{Max-2AND}$ problem is $0.483$-approximable in the streaming setting (whereas the trivial approximation is a $\frac14$-approximation). For more background on CSPs in the streaming model, see the surveys of \textcite{Vel23}, \textcite{Sin22}, and \textcite{Sud22}.

However, this success in obtaining non-trivial algorithms has not extended to any OCSP problem. Indeed, given the known (UGC-)hardness of OCSPs with respect to polynomial time algorithms \cite{GHM+11}, and the empirically-observed phenomenon that subpolynomial space streaming algorithms are linear time simulatable, it would be extremely surprising to find a non-trivial approximation algorithm for an OCSP using subpolynomial space. This work confirms this expectation formally, and unconditionally, by showing that there are no non-trivial sublinear (in $n$) space streaming algorithms for approximating OCSPs.

\subsection{Results}

In this paper, we prove the following theorem:

\begin{theorem}[Main theorem]\label{thm:main}
For every (finite) family of ordering  predicates $\cF$, $\mF$ is approximation-resistant (to single-pass streaming algorithms). In particular, for every $\epsilon > 0$, every $(\rho(\cF) + \epsilon)$-approximation algorithm for $\mF$ requires $\Omega(n)$ space.
\end{theorem}

In particular, for every $\epsilon > 0$, $\mas$ is not $(1/2+\epsilon)$-approximable and $\mbtwn$ is not $(1/3+\epsilon)$-approximable. \ref{thm:main} is proved in \ref{sec:lower-bound}, modulo several technical lemmas proven in later sections.

The linear space bound in \ref{thm:main} is optimal, up to logarithmic factors:

\begin{theorem}[$\tilde{O}(n)$-space algorithm]\label{thm:sparse-alg}
    Let $\cF$ denote any (finite) family of ordering predicates. For all $c > 0$ and $\epsilon > 0$, there exists a single-pass streaming algorithm $\Alg$ which, given an instance $\Psi$ of $\mocsp(\cF)$ with $n$ variables and $m \leq n^c$ constraints, outputs a $(1-\epsilon)$-approximation to $\ocspval_\Psi$ in $O(n \log^3 n/\epsilon^2)$ bits of space.
\end{theorem}

The algorithm in \ref{thm:sparse-alg} is the analogue for OCSPs of a well-known algorithm in the setting of streaming CSPs (see, \cite{KK15,KK19,CGS+22-linear-space}): It simply sparsifies the input instance down to $\tilde{O}(n/\epsilon^2)$ constraints, and then solves the $\mocsp(\cF)$ problem exactly on the sparsified instance. For completeness, we prove \ref{thm:sparse-alg} in \ref{app:sparse-alg}.

\subsection{Related works}

As far as we know, in the streaming setting, \ref{thm:main} is the first tight inapproximability result for $\mF$ for \emph{any} constraint family $\cF$ in $\Omega(n^\delta)$ space for any $\delta > 0$, and it yields tight approximability results for \emph{every} family in \emph{linear} space. 

\ref{thm:main} parallels the classical result of \textcite{GHM+11}, who prove that $\mPi$ is approximation resistant with respect to \emph{polynomial-time} algorithms, for every $\Pi$, assuming the unique games conjecture.\footnote{Without relying on the unique games conjecture, some weaker $\NP$-hardness results are known. For $\mbtwn$, since $\rho(\Pi_{\btwn}) = \frac13$, the trivial algorithm is a $\frac13$-approximation. \textcite{CS98} showed that $(\frac{47}{48}+\epsilon)$-approximating $\mathsf{Max}\btwn$ is $\NP$-hard, for every $\epsilon > 0$. The $\frac{47}{48}$ factor was improved to $\frac12$ by \textcite{AMW15}. For $\mas$, the trivial algorithm is a $\frac12$-approximation. \textcite{New00} showed that $(\frac{65}{66}+\epsilon)$-approximating $\mas$ is $\NP$-hard, for every $\epsilon > 0$. \textcite{AMW15} improved the $\frac{65}{66}$ to $\frac{14}{15}$, and \textcite{BK19} further improved the factor to $\frac23$.} In our setting of streaming algorithms, the only problem that seems to have been previously explored in the literature was $\mas$, and even in this case, a tight approximability result was not known.

In the case of $\mas$, \textcite{GVV17} proved that for every $\epsilon > 0$, $\mas$ is not $(7/8+\epsilon)$-approximable in $o(\sqrt n)$ space using a gadget reduction from the Boolean hidden matching problem \cite{GKK+08}. \textcite{GT19} indicated that $3/4$-approximating $\mas$ is hard in $o(\sqrt n)$ space. Their proof first establishes $o(\sqrt n)$-space approximation resistance for a (non-ordering) CSP called $\mug$ and then reduces from $\mug$ to $\mas$, though this reduction is not fully analyzed.

\textcite{CGMV20} recently also studied directed graph ordering problems (e.g., acyclicity testing, $(s,t)$-connectivity, topological sorting) in the streaming setting. For the problems they consider, they give {\em super-linear} space lower bounds even for multi-pass streaming algorithms. In contrast, as we mentioned above, every OCSP can be approximated arbitrarily well by simple $\tilde{O}(n)$-space algorithms, even in a single pass. However, one of the problems considered in \cite{CGMV20} is close enough to $\mas$ to merit a more detailed comparison: the \emph{minimum feedback arc set} ($\mfas$) problem, the goal of which is to output the fractional size of the smallest set of edges whose removal produces an acyclic subgraph. In other words, the sum of the $\mfas$ value of a graph and the $\mas$ value of the graph is exactly one. The authors of \cite{CGMV20} proved that for every $\kappa > 1$, $\kappa$-approximating\footnote{For minimization problems, a $\kappa$-approximation to a value $v$ is in the interval $[v,\kappa v]$. Thus approximation factors are larger than $1$.} the $\mfas$ value requires $\Omega(n^2)$ space in the streaming setting (for a single pass, and more generally $\Omega(n^{1+\Omega(1/p)}/p^{O(1)})$ space for $p$ passes). We remark that, as is typical for pairs of minimization and maximization problems defined in this way (that is, such that the values sum to 1), the hard instances involved in proving optimal inapproximability are incomparable: In particular, proving $\kappa$-inapproximability for $\mfas$ involves constructing indistinguishable (distributions of) instances with $\mas$ values $\approx 1-\epsilon$ vs. $\approx 1-\kappa \epsilon$ and thus does not imply any hardness of approximation for $\mas$.

A recent work of \textcite{CM23} studies another variant of CSPs, called \emph{phylogenetic CSPs}, where the solution space is a set of \emph{trees},\footnote{Given $n$ variables, the ordering space to a phylogenetic CSP consists of rooted binary, or more generally $k$-ary, trees with $n$ labeled leaves, and constraints specify ``hierarchical structure'' among the leaves. For instance, in the ``triple reconstruction'' problem, the constraint $(a,b,c)$ expresses that $a$ and $b$ are closer to each other than either is to $c$, where distance is measured as path length in the tree. See \cite[\S3,\S9]{CM23} for full definitions.} analogously to how for ordering CSPs, the solution space was a set of permutations. The authors prove approximation resistance of phylogenetic CSPs against polynomial-time algorithms assuming the UGC, and their proof is via value-preserving reductions from ordering CSPs, which were themselves proven approximation-resistant under the UGC in \cite{GHM+11}. It would be interesting to see if the techniques in the current work could help prove streaming approximation lower bounds for phylogenetic CSPs or other CSP variants.

\subsection{Techniques}

Our general approach to prove hardness of $\mocsp$ problems is the following: We choose a family of (standard) CSPs where hardness results are known, and then reduce these CSPs to the OCSPs at hand. While this general approach is not new, in order to achieve optimal streaming hardness results for OCSPs, we need to choose the ``source'' CSPs carefully, so that we can \emph{both} (i) apply previously-known optimal streaming hardness results for these CSPs (in our setting, we use results due to \textcite{CGS+22-linear-space}) and (ii) design \emph{streaming} reductions from these CSPs to $\mocsp$s which produce instances with (almost) optimal ratios in value between $\yes$ and $\no$ instances. In contrast, previous approaches \cite{GVV17,GT19} towards proving hardness of $\mocsp$s (in particular $\mas$) were unable to achieve optimal streaming hardness results \emph{despite} starting with optimal hardness results for the source CSP $\mug$, because of issues in designing streaming reductions which produce sufficiently large value gaps. In the remainder of this section, we describe and motivate this approach towards proving the approximation-resistance of $\mocsp$s.

\subsubsection{Special case: The intuition for $\mas$}

We start by describing our proof technique for the special case of the $\mas$ problem. Let $+_q$ denote the modular addition operator on $[q] = \{1,\ldots,q\}$: For $a,b \in [q]$, $a+_qb$ denotes the unique $c \in [q]$ such that $a+b\equiv c \pmod{q}$. Thus, for instance, $1+_qq=1$.

Similarly to earlier work in the setting of streaming approximability (e.g., the work of \textcite{KKS15}), we prove inapproximability of $\mas$ by exhibiting a pair of distributions over $\mas$ instances, which we denote $\cY$ (the ``$\yes$ instances'') and $\cN$ (the ``$\no$ instances''), satisfying the following two properties:
\begin{enumerate}
    \item $\cY$ and $\cN$ are ``indistinguishable'' to streaming algorithms (in a sense we define formally below).
    \item With high probability, $\cY$ has high $\mas$ values ($\approx 1$) and $\cN$ has low $\mas$ values ($\approx \frac12$).
\end{enumerate}
The existence of such distributions would suffice to establish the theorem: there cannot be any streaming approximation for $\mas$, since any such algorithm would be able to distinguish these distributions. But how are we to construct distributions $\cY$ and $\cN$ satisfying these properties?

The ``recipe'' which has proved successful in past works for proving streaming approximation resistance for \emph{``standard''} CSPs is roughly to let the $\cN$ instances be completely random, while $\cY$ instances are sampled with ``hidden structure'' which guarantees a very good assignment. Then, one would show that streaming algorithms cannot detect the existence of such hidden structure, via a reduction to a communication game (typically a variant of Boolean hidden matching \cite{GKK+08,VY11}). In the OCSP setting, we might hope that the hidden structure could simply be an ordering; that is, we could hope to define $\cY$ by first sampling a random ordering of the variables, then sampling constraints that go forward with respect to this ordering, and then perhaps adding some noise. But unfortunately, we don't know how to directly prove communication lower bounds for such problems.

Hence, instead of seeking to prove indistinguishability directly, we turn back to earlier streaming hardness-of-approximation results proven in the context of standard CSPs. In this setting, variables take on values in a finite alphabet $[q]$ (i.e., the solution space is $[q]^n$), and $k$-ary predicates $f : [q]^k \to \{0,1\}$ can be applied to small subsets of variables to form constraints. We make two observations about this definition. Firstly, in a CSP, two variables may be assigned the same value in $[q]$, whereas in an OCSP, every variable must get a distinct value in $[n]$. Secondly, for a CSP or OCSP defined by a single binary predicate, each constraint simply specifies a pair $(j_1,j_2)$ of distinct indices in $[n]$; by extension, instances can be viewed equivalently as directed graphs on $[n]$ (allowing multiple edges). Thus, we can view instances of binary CSPs as instances of $\mas$, and vice versa.

The plan is as follows. We'll define a binary predicate denoted $\Pi_{\mas}^{\qcoar} : [q]^2 \to \{0,1\}$. Let $\mcsp(\Pi_{\mas}^{\qcoar})$ denote the problem of maximizing $\Pi_{\mas}^{\qcoar}$ constraints applied to assignments in $[q]^n$. The hope is that for a careful choice of the alphabet size $q$ and the predicate $\Pi_{\mas}^{\qcoar}$, we can \emph{reuse} indistinguishable $\yes$/$\no$ distributions for $\mcsp(\Pi_{\mas}^{\qcoar})$ --- in particular, those constructed in the recent work of \textcite{CGS+22-linear-space} --- as $\yes$/$\no$ distributions for $\mas$. This requires us to relate the values of an $\mas$ instance and the corresponding $\mcsp(\Pi_{\mas}^{\qcoar})$ instance. To be precise, for an $\mas$ instance $\Psi$, let $\Psiq$ denote the $\mcsp(\Pi_{\mas}^{\qcoar})$ instance with the exact same list of constraints and let $\cspval_{\Psiq}$ denote the value of this instance. We choose $q$ and $\Pi_{\mas}^{\qcoar}$ so as to imply the following four properties about the indistinguishable distributions $\cY$ and $\cN$ ``given to us'' by \cite{CGS+22-linear-space}:

\begin{enumerate}
    \item With high probability over $\Psi \sim \cY$, $\cspval_{\Psiq} \approx 1$.\label{item:coarse-val-y}
    \item With high probability over ${\Psi \sim \cN}$, $\cspval_{\Psiq} \approx \frac12$.\label{item:coarse-val-n}
    \item For all $\Psi$, $\ocspval_\Psi \geq \cspval_{\Psiq}$.\label{item:gap-y}
    \item With high probability over ${\Psi \sim \cN}$, $\ocspval_\Psi$ is not much larger than $\cspval_{\Psiq}$.\label{item:gap-n}
\end{enumerate}

Together, these items will suffice to prove the theorem since \ref{item:coarse-val-n} and \ref{item:gap-n} together imply that with high probability over $\Psi \sim \cN$, $\ocspval_\Psi \approx \frac12$, while \ref{item:coarse-val-y} and \ref{item:gap-y} together imply that with high probability over $\Psi \sim \cY$, $\ocspval_\Psi \approx 1$.

In order to satisfy these criteria, we define the CSP predicate as follows. Recall that $\Pi_{\mas}(1,2)= 1$ while $\Pi_{\mas}(2,1)=0$. We define the constraint predicate $\Piq_{\mas} : [q]^2 \to \{0,1\}$ by $\Piq_{\mas}(b_1,b_2) = 1$ iff $b_1 < b_2$. We call this the \emph{$q$-coarsening} of $\Pi_{\mas}$, and it gives $\mcsp(\Pi_{\mas}^{\qcoar})$ the following ``scheduling'' interpretation. Recall, the goal of $\mas$ is to schedule $n$ tasks, each task $i$ is assigned a distinct position $\sigma_i \in [n]$ in the schedule, and the goal is to maximize constraints of the form $(j_1,j_2)$ requiring that task $j_1$ takes place before task $j_2$, i.e., $\sigma_{j_1} < \sigma_{j_2}$. In $\mcsp(\Pi_{\mas}^{\qcoar})$, the goal is to schedule $n$ tasks in $q$ \emph{batches}: Each task $i$ receives a (not necessarily distinct!) batch $\sigma_i \in [q]$, and constraints $(j_1,j_2)$ still require that $\sigma_{j_1} < \sigma_{j_2}$, that is, $j_1$'s assigned batch is earlier than $j_2$. \ref{item:gap-y} follows immediately in this interpretation: Given any batched schedule $\vecb \in [q]^n$, we can immediately ``lift'' to a non-batched schedule $\vecb^{\refn} \in \sym_n$ by arbitrarily ordering the tasks in each batch, which can only increase the number of satisfied constraints.

Proving \ref{item:gap-n} is the meat of the argument. Note that if we set $q=n$, $\mcsp(\Pi_{\mas}^{\qcoar})$ becomes the same problem as $\mas$, and hence \ref{item:gap-n} is trivial! However, we can only apply the inapproximability results of \cite{CGS+22-linear-space} (specifically, the indistinguishability of their distributions $\cY$ and $\cN$) when $q$ is a constant. Briefly, the \cite{CGS+22-linear-space} results roughly state that a predicate $f : [q]^k \to \{0,1\}$ is inapproximable when its support satisfies a property that they call \emph{width}: it contains most of a ``diagonal'', in the sense that for some $\veca \in [q]^k$, the set $\{c \in [q] : f(\veca +_q (c,\ldots,c)) = 1\}$ is large. Luckily for us, $\Pi_{\mas}^{\qcoar}$ has this property with $\veca = (1,2)$; indeed, $\Pi_{\mas}^{\qcoar}(1+_qc,2+_qc) = 1$ unless $c = q-1$ (in which case $1+_qc = q$ while $2+_qc = 1$).

To actually prove \ref{item:gap-n}, then, we can no longer use the results of \cite{CGS+22-linear-space} as a black box. Specifically, we need to understand the structure of the $\no$ distribution $\cN$ (beyond \ref{item:coarse-val-n} and its indistinguishability from $\cY$). We show that instances drawn from $\cN$ are (with high probability) ``small partition expanders'' in a specific sense: for every partition of the set of variables into $q$ blocks of roughly equal size, very few constraints, specifically a $o(1)$ fraction, involve two variables in the same block. (See \ref{def:sphe}.) Now, we think of a ``schedule'' $\vecsigma \in \sym_n$ as giving rise to a ``batched schedule'' $\vecsigma^{\qcoar} \in [q]^n$ in the following way: If task $i$ is scheduled in position $\sigma_i \in [n]$, then we place it in batch $\approx \sigma_i q / n$. Thus the first $\approx n/q$ scheduled tasks are placed in batch $1$, the next $\approx n/q$ in batch $2$, etc. Hence, whenever a constraint $(j_1,j_2)$ is satisfied by $\vecsigma$ (as an $\mas$ constraint), it will also be satisfied by $\vecsigma^{\qcoar}$ (as a $\mcsp(\Pi_{\mas}^{\qcoar})$ constraint), \emph{unless} $j_1$ and $j_2$ end up in the same batch; but by the small partition expansion condition, this happens only for $o(1)$ fraction of the constraints. Hence $\ocspval_{\Psi} \leq \cspval_{\Psi^{\qcoar}} + o(1)$.

\subsubsection{Extending to general ordering CSPs}

Extending the idea to other OCSPs follows the same basic outline. Given the constraint predicate $\Pi : \sym_k \to \{0,1\}$ (of arity $k$) and positive integer $q$, we define $\Piq : [q]^k \to \{0,1\}$ analogously to $\Piq_\mas$: $\Piq(\vecb)$ is $\Pi(\ord(\vecb))$ if $\ord(\vecb) \neq \bot$ (i.e., $\vecb$'s entries are all distinct), and $0$ otherwise. We then describe the $\yes$ and $\no$ distributions of $\mPiq$ which the general theorem of \cite{CGS+22-linear-space} shows are indistinguishable to $o(n)$ space algorithms, again taking advantage of the fact that $\Pi^q$'s support mostly contains a diagonal. Finally, we give an analysis of the partition expansion in the $\no$ instances arising from the construction in~\cite{CGS+22-linear-space}. Specifically, we show that the instances are now a ``small partition hypergraph expander'', in the sense that for every partition of the $n$ variables into $q$ blocks of roughly equal size, very few constraints involve even two vertices from the same block.
 
\subsubsection{Further remarks}

Our notion of coarsening is somewhat similar to, but not the same as, that used in previous works, notably~\cite{GHM+11}. In particular, the techniques used to compare the OCSP value (before coarsening) with the standard CSP value (after coarsening) are somewhat different: Their analysis involves more sophisticated tools such as influence of variables and Gaussian noise stability. The proof in our setting, in contrast, uses a more elementary analysis of the type common with random graphs.

In the rest of the paper, in the interest of self-containedness, we will avoid invoking the work of \cite{CGS+22-linear-space} on linear-space streaming CSP  inapproximability where possible. Instead, we will explicitly define the distributions $\cY$ and $\cN$ over $\mPi$ instances for arbitrary ordering predicates $\Pi$ and analyze them directly, without invoking any prior analyses of their coarsened CSP values (which would require formally defining the notion of the ``width'' of predicates in \cite{CGS+22-linear-space}). Hence, we'll only need to invoke \cite{CGS+22-linear-space} in the context of using communication lower bounds to prove indistinguishability of $\cY$ and $\cN$. We also manage to prove a stronger statement about the coarsened $\cY$ distribution (though it is unnecessary for our application): Its value is high with probability $1$, as opposed to just $1-o(1)$ (which would be implied by the analysis in \cite{CGS+22-linear-space}).

\paragraph{Organization of the rest of the paper.} In \ref{sec:prelim} we introduce some additional notation and background material. In \ref{sec:lower-bound}, we introduce two distributions on $\mPi$ instances, the \yes\ distribution $\cY$ and the \no\ distribution $\cN$; state lemmas asserting that these distributions are concentrated on instances with high, and respectively low, OCSP value; and that these distributions are indistinguishable to (single-pass) small-space streaming algorithms; and then prove \ref{thm:main} modulo these lemmas. Finally, we prove the lemmas on the OCSP values in \ref{sec:bounds}, and prove the indistinguishability lemma in \ref{sec:streaming}.

\paragraph{Acknowledgements.} A previous version of this paper appeared in APPROX 2021~\cite{SSV21-conf-version}. The results starting from that version improve on an earlier version of the paper~\cite{SSV21-early-version} that gave only $\Omega(\sqrt{n})$ space lower bounds for all OCSPs. Our improvement to $\Omega(n)$ space lower bounds comes by invoking the more recent results of \textcite{CGS+22-linear-space}, whereas our previous version used the strongest lower bounds for CSPs that were available at the time from an earlier work of \textcite{CGSV24}.\footnote{The conference version of \cite{CGSV24} appeared in 2021. The results in \cite{CGSV24} are quantitatively weaker for the problems considered in \cite{CGS+22-linear-space}, though their results apply to a broader collection of problems. Interestingly, for our application, which covers {\em all} OCSPs, the narrower set of problems considered in \cite{CGS+22-linear-space} suffices.}

We would like to thank the anonymous referees at \emph{Computational Complexity} and APPROX for their helpful comments.\vspace{0.1in}

\noindent \textsc{n.g.s.} was supported by an NSF Graduate Research Fellowship (Award DGE2140739). Work was done in part when the author was an undergraduate student at Harvard University.

\noindent \textsc{m.s.} was supported in part by a Simons 
Investigator Award and NSF Award CCF 1715187. 

\noindent \textsc{s.v.} was supported in part by a Google Ph.D. Fellowship, a Simons Investigator Award to Madhu Sudan, and NSF Award CCF 2152413. Work was done in part when the author was a graduate student at Harvard University.

\section{Preliminaries and definitions}\label{sec:prelim}

\subsection{Additional notation}

The \emph{support} of an ordering constraint predicate $\Pi : \sym_k \to \{0,1\}$ is the set $\supp(\Pi) = \{\vecpi \in \sym_k : \Pi(\vecpi) =1\}$.

We first define a notion of ``$k$-hypergraphs''. (These are $k$-uniform ordered hypergraphs with multiple hyperedges and without self-loops.)
Given a finite set $V$, an (ordered, self-loop-free) {\em $k$-hyperedge} $\vecj = (j_1,\ldots,j_k)$ is a sequence of $k$ distinct elements $j_1,\ldots,j_k \in V$. We stress that the ordering of vertices within an edge is important to us. 
An \emph{$k$-hypergraph} $G = (V,E)$ is given by a set of vertices $V$ and a multiset $E \subseteq V^k$ of $k$-hyperedges on $V$. A $k$-hyperedge $\vecj$ is \emph{incident} on a vertex $v$ if $v$ appears in $\vecj$. Let $\Gamma(\vecj) \subseteq V$ denote the set of vertices to which a $k$-hyperedge $\vecj$ is incident, and let $m = m(G)$ denote the number of $k$-hyperedges in $G$.

A $k$-hypergraph is a \emph{$k$-hypermatching} if it has the property that no pair of (distinct) $k$-hyperedges is incident on the same vertex. We let $\matchings_{k,\alpha}(n)$ denote the uniform distribution over all $k$-hypermatchings on $[n]$ with $\alpha n$ edges.

A vector $\vecb = (b_1,\ldots,b_n) \in [q]^n$ may be viewed as a \emph{$q$-partition} of $[n]$ into \emph{blocks} $\vecb^{-1}(1),\ldots,\vecb^{-1}(q)$, where the $i$-th block $\vecb^{-1}(i)$ is defined as the set of indices $\{j \in [n] : b_j = i\}$. Given $\vecb = (b_1,\ldots,b_n) \in [q]^n$ and an indexing vector $\vecj = (j_1,\ldots,j_k) \in [n]^k$, we define $\vecb|_{\vecj} = (b_{j_1},\ldots,b_{j_k}) \in [q]^k$.

Given an instance $\Psi$ of $\mPi$ on $n$ variables, we define its \emph{constraint hypergraph} $G(\Psi)$ to be the $k$-hypergraph on $[n]$ consisting of the $k$-hyperedge $\vecj$ for each constraint $(\Pi,\vecj)$ in $\Psi$. We also let $m(\Psi)$ denote the number of constraints in $\Psi$ (equiv., the number of $k$-hyperedges in $G(\Psi)$).

\subsection{Concentration bounds}

We require the following Azuma-style concentration inequality for (not necessarily independent) Bernoulli random variables with bounded conditional expectations taken from \textcite{KK19}:

\begin{lemma}[{{\cite[Lemma~2.5]{KK19}}}]\label{lemma:azuma}
Let $0 < p < 1$. Let $X_1,\ldots,X_m$ be $\{0,1\}$-valued random variables such that for every $i \in [m]$, $\Exp[X_i\mid X_1,\ldots,X_{i-1}] \leq p$. Then for every $\eta > 0$, \[ \Pr\left [\sum_{i=1}^m X_i \geq (p+\eta) m \right] \leq \exp\left(-\left(\frac{\eta^2}{2(p+\eta)}\right) m\right). \]
\end{lemma}

We also require standard Chernoff bounds for sums of independent Bernoulli variables which we state here for completeness:

\begin{lemma}[Chernoff bounds]\label{lemma:chernoff}
Let $0 < p < 1$. Let $X_1,\ldots,X_m$ be $\{0,1\}$-valued random variables such that for every $i \in [m]$, $\Exp[X_i] = p$. Then:
\begin{enumerate}
    \item For every $\eta > 0$, \[ \Pr\left[ \sum_{i=1}^m X_i \leq (p-\eta) m\right] \leq \exp\left(- \left(\frac{\eta^2}{2p}\right) m\right). \]\label{item:chernoff:lower}
    \item For all $\eta > 0$, \[ \Pr\left[ \left\lvert \sum_{i=1}^m X_i - pm\right\rvert \geq \eta m \right] \leq 2\exp\left(- \left(\frac{\eta^2}{3p}\right) m\right). \]\label{item:chernoff:twosided}
\end{enumerate}
\end{lemma}

(Note that the lower bounds in these lemmas are trivial if $\eta > p$.)

\subsection{Stirling's approximation}

Finally, we state a standard form of Stirling's bound for the factorial:

\begin{lemma}[Stirling approximation]\label{lem:stirling}
    For all $n \in N$, \[ \sqrt{2\pi n} (n/e)^n < n! < 2 \sqrt{2\pi n} (n/e)^n. \]
\end{lemma}


\section{The streaming space lower bound}\label{sec:lower-bound}

In this section, we prove our main theorem (\ref{thm:main}), modulo some lemmas that we prove in later sections. We focus first on the following special case for single-predicate families:

\begin{theorem}[Main theorem (single-predicate case)]\label{thm:main-single}
For every $k \in \N$ and every predicate $\Pi : \sym_k \to \{0,1\}$, $\mPi$ is approximation-resistant (to single-pass streaming algorithms). In particular, for every $\epsilon > 0$, every $(\rho(\Pi) + \epsilon)$-approximation algorithm for $\mPi$ requires $\Omega(n)$ space.
\end{theorem}

Indeed, given \ref{thm:main-single}, \ref{thm:main} follows immediately:

\begin{proof}[Proof of \ref{thm:main}]
Given any family $\cF$ of predicates, let $\Pi$ have minimal random assignment value $\rho$ over predicates in $\cF$, so that $\rho(\Pi) = \rho(\cF)$. Then since every instance of $\mPi$ is also an instance of $\mF$, \ref{thm:main-single} immediately implies $(\rho(\cF)+\epsilon)$-approximations for $\mF$ require $\Omega(n)$ space.
\end{proof}

Our lower bound is proved, as is usual for such statements, by showing that no small space algorithm can ``distinguish'' $\yes$ instances with OCSP value at least $1-\epsilon/2$, from $\no$ instances with OCSP value at most $\rho(\Pi)+\epsilon/2$. Such a statement is in turn proved by exhibiting two distributions, the $\yes$ distribution $\cY$ and the $\no$ distribution $\cN$, and showing these are indistinguishable. Specifically, we carefully choose some parameters $q, T, \alpha$ and a permutation $\vecpi \in \sym_k$, and define two distributions $\cY = \DY$ and $\cN = \DN$ over $n$-variable instances of $\mPi$.
We claim that for our choice of parameters $\cY$ is supported on instances with value at least $1 - \epsilon/2$ --- this is asserted in \ref{lemma:y-lower-bound} below. 
Similarly, we claim that $\cN$ is mostly supported (with probability $1-o(1)$) on instances with value at most $\rho(\Pi)+\epsilon/2$ (see \ref{lemma:n-upper-bound}). 
Finally, we assert in \ref{lem:our-indist} that any algorithm that distinguishes $\cY$ from $\cN$ with ``advantage'' at least $1/8$ (i.e., accepts $\Psi\sim\cY$ with probability $1/8$ more than $\Psi \sim \cN$) requires $\Omega(n)$ space. 

Assuming \ref{lemma:y-lower-bound}, \ref{lemma:n-upper-bound}, and \ref{lem:our-indist}, the proof of \ref{thm:main-single} is straightforward and given at the end of this section. We prove \ref{lemma:y-lower-bound} and \ref{lemma:n-upper-bound} in \ref{sec:bounds} and \ref{lem:our-indist} in \ref{sec:streaming}.

\subsection{Distribution of hard instances}\label{sec:hard_distributions}

We now formally define our $\yes$ and $\no$ distributions for $\mPi$.

\begin{definition}[$\DY$ and $\DN$]\label{def:YES_NO_MaxOCSP_dist}
For $k \in \N$ and $\Pi : \sym_k \to \{0,1\}$, let $q,n,T \in \N$ and $\alpha > 0$, $q \geq k$, and let $\vecpi \in \supp(\Pi)$. We define two distributions over $\mPi$ instances with $n$ variables, the $\yes$ distribution $\DY$ and the $\no$ distribution $\DN$, as follows:
\begin{enumerate}
    \item Sample a uniformly random $q$-partition $\vecb = (b_1,\ldots,b_n) \in [q]^n$.
    \item Sample $T$ hypermatchings $\tilde G_1,\ldots,\tilde G_T \sim \matchings_{k,\alpha}(n)$ independently.
    \item For each $t \in [T]$, do the following:
    \begin{itemize}
        \item Let $G_t$ be an empty $k$-hypergraph on $[n]$.
        \item For each $k$-hyperedge $\vecj = (j_1,\ldots,j_k) \in E(\tilde G_t)$:
        \begin{itemize}
            \item $\yes$ case: If there exists $c \in [q]$ such that $\vecb|_\vecj = \vecpi +_q (c,\ldots,c)$, add $\vecj$ to $G_t$ with probability $1/q$. (Here $\vecpi$ is viewed as $k$-tuple in $[k]^k\subseteq[q]^k$.)
            \item $\no$ case: Add $\vecj$ to $G_t$ with probability $\frac1{q^k}$.
        \end{itemize}
    \end{itemize}
    \item Set $G \gets G_1 \cup \cdots \cup G_T$.
    \item Return the $\mPi$ instance $\Psi$ on $n$ variables given by the constraint hypergraph $G$.
\end{enumerate}
\end{definition}

We say that an algorithm $\Alg$ achieves \emph{advantage} $\delta$ in distinguishing $\DY$ from $\DN$ if there exists an $n_0$ such that for all $n \geq n_0$, we have 
$$\left\lvert\Pr_{\Psi \sim \DY}[\Alg(\Psi)=1] - \Pr_{\Psi \sim \DN}[\Alg(\Psi)=1]\right\rvert \geq \delta.$$

We make several remarks on this definition. Firstly, note that the constraints within $\DY$ and $\DN$ do not directly depend on $\Pi$. We still parameterize the distributions by $\Pi$, since they are formally distributions over $\mPi$ instances; $\Pi$ also determines the arity $k$ and the set of allowed permutations $\vecpi$ in the $\yes$ case. Secondly, we note that when sampling an instance from $\DN$, the partition $\vecb$ is ignored, and so $\DN$ is ``random''. Hence these instances fit into the typical streaming lower bound ``recipe'' of ``random graphs vs. random graphs with hidden structure''. Finally, we observe that the number of constraints in both distributions is distributed as a sum of $m = n \alpha T$ independent Bernoulli$(\frac1{q^k})$ random variables.

In the following section, we state lemmas which highlight the main properties of the distributions above. See \ref{fig:mas} for a visual interpretation of the distributions in the case of $\mas$.

\begin{figure}
    \centering
    \begin{subfigure}{\textwidth}
        \centering
\begin{tikzpicture}[scale=0.75,vertex/.style={fill=black},block/.style={draw=black,fill=white!70!lightgray},edge/.style={->,line width=1.5pt,-{Latex[width=8pt,length=10pt]}}]
\draw[block] (-.4,-.4) rectangle (1.4,1.4);
\node at (-.2,1.7) {$\mathbf{1}$};
\draw[vertex] (0,0) circle (3pt);
\draw[vertex] (0,1) circle (3pt);
\draw[vertex] (1,0) circle (3pt);
\draw[vertex] (1,1) circle (3pt);
\draw[block] (2.6,.1) rectangle (4.4,.9);
\node at (2.8,1.2) {$\mathbf{2}$};
\draw[vertex] (3,.5) circle (3pt);
\draw[vertex] (4,.5) circle (3pt);
\draw[block] (5.6,.1) rectangle (6.4,.9);
\node at (5.8,1.2) {$\mathbf{3}$};
\draw[vertex] (6,.5) circle (3pt);
\draw[block] (7.6,-0.4) rectangle (9.4,1.4);
\node at (7.8,1.7) {$\mathbf{4}$};
\draw[vertex] (8.5,1) circle (3pt);
\draw[vertex] (8,0) circle (3pt);
\draw[vertex] (9,0) circle (3pt);
\draw[block] (10.6,.1) rectangle (12.4,.9);
\node at (10.8,1.2) {$\mathbf{5}$};
\draw[vertex] (11,.5) circle (3pt);
\draw[vertex] (12,.5) circle (3pt);
\draw[edge,draw=black!40!green] (1,1) to[bend left=55] (4,.5);
\draw[edge,draw=black!40!green] (4,.5) to[bend left=40] (6,.5);
\draw[edge,draw=black!40!green] (3,.5) to[bend right=40] (6,.5);
\draw[edge,draw=black!40!green] (6,.5) to[bend left=40] (8.5,1);
\draw[edge,draw=black!40!green] (6,.5) to[bend right=40] (8,0);
\draw[edge,draw=black!40!green] (8,0) to[bend right=40] (11,.5);
\draw[edge,draw=black!40!green] (9,0) to[bend left=10] (11,.5);
\draw[edge,draw=black!10!red] (12,.5) to[bend left=30] (0,0);
\draw[edge,draw=black!10!red] (11,.5) to[bend right=35] (0,1);
\end{tikzpicture}
    \caption{Constraint graph of a sample $\mas$ instance drawn from $\cY$}
    \end{subfigure}
    \begin{subfigure}{\textwidth}
    \centering
\begin{tikzpicture}[scale=0.75,vertex/.style={fill=black},block/.style={draw=black,fill=white!70!lightgray},edge/.style={->,line width=1.5pt,-{Latex[width=8pt,length=10pt]}}]
\draw[block] (-.4,-.4) rectangle (1.4,1.4);
\node at (-.2,1.7) {$\mathbf{1}$};
\draw[vertex] (0,0) circle (3pt);
\draw[vertex] (0,1) circle (3pt);
\draw[vertex] (1,0) circle (3pt);
\draw[vertex] (1,1) circle (3pt);
\draw[block] (2.6,.1) rectangle (4.4,.9);
\node at (2.8,1.2) {$\mathbf{2}$};
\draw[vertex] (3,.5) circle (3pt);
\draw[vertex] (4,.5) circle (3pt);
\draw[block] (5.6,.1) rectangle (6.4,.9);
\node at (5.8,1.2) {$\mathbf{3}$};
\draw[vertex] (6,.5) circle (3pt);
\draw[block] (7.6,-0.4) rectangle (9.4,1.4);
\node at (7.8,1.7) {$\mathbf{4}$};
\draw[vertex] (8.5,1) circle (3pt);
\draw[vertex] (8,0) circle (3pt);
\draw[vertex] (9,0) circle (3pt);
\draw[block] (10.6,.1) rectangle (12.4,.9);
\node at (10.8,1.2) {$\mathbf{5}$};
\draw[vertex] (11,.5) circle (3pt);
\draw[vertex] (12,.5) circle (3pt);
\draw[edge,draw=black!40!green] (1,1) to[bend left=55] (4,.5);
\draw[edge,draw=black!40!green] (3,.5) to[bend left=55] (11,.5);
\draw[edge,draw=black!40!green] (1,0) to[bend right=45] (6,.5);
\draw[edge,draw=black!40!green] (6,.5) to[bend right=20] (8.5,1);
\draw[edge,draw=orange] (8,0) to[bend right=40] (9,0);
\draw[edge,draw=orange] (4,.5) to[bend left=40] (3,.5);
\draw[edge,draw=black!10!red] (6,.5) to[bend right=40] (0,1);
\draw[edge,draw=black!10!red] (8.5,1) to[bend right=20] (6,.5);
\draw[edge,draw=black!10!red] (11,.5) to[bend left=35] (0,0);
\end{tikzpicture}
    \caption{Constraint graph of a sample $\mas$ instance drawn from $\cN$}
    \end{subfigure}
    \caption[]{The constraint graphs of $\mas$ instances which could plausibly be drawn from $\cY$ and $\cN$, respectively, for $q=5$ and $n=12$. Recall that $\mas$ is a binary $\mocsp$ with ordering constraint function $\Pi_{\mas}$ supported only on $(1,2)$. According to the definition of $\cY$ (see \ref{def:YES_NO_MaxOCSP_dist}, with $\vecpi = (1,2)$), instances are sampled by first sampling a $q$-partition $\vecb = (b_1,\ldots,b_n) \in [q]^n$, and then sampling some constraints; every sampled constraint $(j_1,j_2)$ must satisfy $b_{j_2} \equiv b_{j_1}+1 \pmod q$. On the other hand, there are no requirements on $(b_{j_1},b_{j_2})$ for instances sampled from $\cN$. Above, the blocks of the partition $\vecb$ are labeled $1,\ldots,5$, and the reader can verify that the edges satisfy the appropriate requirements. We also color the edges in a specific way: We color an edge $(j_1,j_2)$ green, orange, or red if $b_{j_2} > b_{j_1}$, $b_{j_2} = b_{j_1}$, or $b_{j_2} < b_{j_1}$, respectively. This visually suggests important elements of our proofs that $\cY$ has $\mas$ values close to $1$ and $\cN$ has $\mas$ values close to $\frac12$ (for formal statements, see \ref{lemma:y-lower-bound} and \ref{lemma:n-upper-bound}, respectively). Specifically, in the case of $\cY$, if we arbitrarily arrange the vertices in each block, we will get an ordering in which every green edge is satisfied, and we expect all but $\frac1q$ fraction of the edges to be green (i.e., all but those which go from block $q$ to block $1$). On the other hand, if we executed a similar process in $\cN$, the resulting ordering would satisfy all green edges and some subset of the orange edges; however, in expectation, these account only for $\frac{q(q+1)}{2q^2} = \frac{q+1}{2q} \approx \frac12$ fraction of the edges.
    }
    \label{fig:mas}
\end{figure}

\subsection{Statement of key lemmas}

Our first lemma shows that $\cY$ is supported on instances of high value.

\begin{lemma}[$\cY$ has high $\mPi$ values]\label{lemma:y-lower-bound}
For every ordering constraint predicate $\Pi : \sym_k \to \{0,1\}$, every $\vecpi \in \supp(\Pi)$ and $\Psi\sim \DY$, we have $\ocspval_\Psi \geq 1-\frac{k-1}{q}$ (i.e., this occurs with probability 1).
\end{lemma}

We prove \ref{lemma:y-lower-bound} in \ref{sec:y-lower-bound}. Next, we assert that $\cN$ is supported mostly on instances of low value. 

\begin{lemma}[$\cN$ has low $\mPi$ values]\label{lemma:n-upper-bound}
For every ordering constraint predicate $\Pi : \sym_k \to \{0,1\}$, and every $\epsilon >0$, there exists $q_0 \in \N$ and $\alpha_0 \geq 0$ such that for all $q \geq q_0$ and $\alpha \leq \alpha_0$, there exists $T_0 \in \N$ such that for all $T \geq T_0$, for sufficiently large $n$, we have \[ \Pr_{\Psi \sim \DN}\left[\ocspval_\Psi \geq \rho(\Pi) + \frac{\epsilon}2\right] \leq 0.01. \]
\end{lemma}

We prove \ref{lemma:n-upper-bound} in \ref{sec:n-upper-bound}. We note that this lemma is more technically involved than 
\ref{lemma:y-lower-bound} and this is the proof that needs the notion of ``small partition expanders''.
Finally, the following lemma asserts the indistinguishability of $\cY$ and $\cN$ to small space streaming algorithms. We remark that this lemma follows directly from the work of \cite{CGS+22-linear-space}, but we prove it for completeness in \ref{sec:streaming}.

\begin{lemma}[$\cY$ and $\cN$ are indistinguishable]\label{lem:our-indist}
For every $q \geq k \in \N$ there exists  $\alpha_0(k)>0$ such that for every
$T\in\N$, $\alpha\in(0,\alpha_0(k)]$ the following holds: 
For every ordering constraint predicate $\Pi :\sym_k \to \{0,1\}$ and $\vecpi\in\supp(\Pi)$, every streaming algorithm distinguishing
$\DY$ from  $\DN$ with advantage $1/8$ for all lengths $n$ uses space $\Omega(n)$.
\end{lemma}

\subsection{Proof of \ref{thm:main-single}} 

We now prove \ref{thm:main-single}.

\begin{proof}[Proof of \ref{thm:main-single}]
Let $\Alg$ be an algorithm for $\mPi$ that uses space $s(n)$ and achieves a $(\rho(\Pi)+\epsilon)$-approximation. Fix $\vecpi \in \supp(\Pi)$. Consider the algorithm $\Alg'$ defined as follows: on input $\Psi$, an instance of $\mPi$, if $\Alg(\Psi)\ge \rho(\Pi)+\frac{\epsilon}{2}$, then $\Alg'$ outputs $1$, else, it outputs $0$. Observe that $\Alg'$ uses $O(s(n))$ space. Set $q_0\ge\frac{2(k-1)}{\epsilon}$ such that the condition of \ref{lemma:n-upper-bound} holds. Set $\alpha_0\in (0,\alpha_0(k)]$ such that the conditions of \ref{lemma:n-upper-bound} holds. Consider any $q\ge q_0$ and $\alpha\le \alpha_0$: let $T_0$ be set as in \ref{lemma:n-upper-bound}. Consider any $T\ge T_0$: since $q\ge \frac{2(k-1)}{\epsilon} $, it follows from \ref{lemma:y-lower-bound} that for $\Psi \sim \DY$, we have $\ocspval_{\Psi} \ge 1-\frac{\epsilon}{2}$, and hence with probability at least $2/3$, $\Alg(\Psi)\ge \rho(\Pi)+\frac{\epsilon}{2}$. Therefore, $\Exp_{\Psi \sim \DY}[\Alg(\Psi)=1]\ge 2/3$. Similarly, by the choice of $q_0,\alpha_0,T_0$, it follows from \ref{lemma:n-upper-bound} that \[ \Pr_{\Psi \sim \DN}\left[\ocspval_\Psi \geq \rho(\Pi) + \frac{\epsilon}2\right] \leq 0.01, \]and hence, $\Exp_{\Psi \sim \DN}[\Alg(\Psi)=1] \le \frac{1}{3} + 0.01$. Therefore, $\Alg'$ distinguishes $\DY$ from  $\DN$ with advantage $1/8$. By applying \ref{lem:our-indist}, we conclude that $\Alg$ uses $s(n) \geq \Omega(n)$ space.
\end{proof}

\section{Bounds on $\mPi$ values of $\cY$ and $\cN$}\label{sec:bounds}

The goal of this section is to prove our technical lemmas \ref{lemma:y-lower-bound} and \ref{lemma:n-upper-bound} which, respectively, lower bound the $\mPi$ values of $\DY$ and upper bound the $\mPi$ values of $\DN$.

\subsection{CSPs and coarsening}\label{sec:standard-csps}

In order to prove the lemmas, we recall the definition of (standard) \emph{constraint satisfaction problems (CSPs)}, whose solution spaces are $[q]^n$ (as opposed to $\sym_n$ for OCSPs), and define an operation called \emph{$q$-coarsening} on $\mocsp$s, which restricts the solution space from $\sym_n$ to $[q]^n$.

A \emph{maximum constraint satisfaction problem}, $\mcsp(f)$, is specified by a single constraint predicate $f : [q]^k \rightarrow \{0,1\}$, for some positive integer $k$. An {\em instance} of $\mcsp(f)$ on $n$ variables is given by $m$ constraints $C_1,\ldots,C_m$ where $C_i = (f,\vecj(i))$, i.e., the application of the predicate $f$ to the variables $\vecj(i) = (j(i)_1,\ldots,j(i)_k)$. (Again, $f$ is omitted when clear from context.) The \emph{value} of an assignment $\vecb \in [q]^n$ on an instance $\Phi = (C_1,\ldots,C_m)$, denoted $\cspval_\Phi(\vecb)$, is the fraction of constraints satisfied by $\vecb$, i.e., $\cspval_\Phi(\vecb)=\tfrac{1}{m}\sum_{i\in[m]} f(\vecb|_{\vecj(i)})$, where (recall) $\vecb|_{\vecj} = (b_{j_1},\ldots,b_{j_k})$ for $\vecb = (b_1,\ldots,b_n), \vecj = (j_1,\ldots,j_k)$. The optimal value of $\Phi$ is defined as $\cspval_\Phi = \max_{\vecb \in [q]^n}\{\cspval_\Phi(\vecb)\}$.

\begin{definition}[$q$-coarsening]\label{defn:coarsening}
Let $\Pi : \sym_k \to \{0,1\}$ be an ordering predicate. For $q \in \N$, the \emph{$q$-coarsening} of $\Pi$ is the predicate $\Piq: [q]^k \to \{0,1\}$ defined by $\Piq(\veca) = 1$ iff $\Pi(\ord(\veca))=1$ (if $\ord(\veca) = \bot$, i.e., $\veca$'s entries are not all distinct, we define $\Piq(\veca) = 0$). The \emph{$q$-coarsening of the problem} $\mPi$ is the problem $\mPiq$, and the \emph{$q$-coarsening of an instance} $\Psi$ of $\mPi$ is the instance $\Psiq$ of $\mPiq$ given by the same constraint hypergraph.
\end{definition}

The following lemma captures the idea that coarsening restricts the space of possible solutions; compare to \ref{lemma:sphe-gap-bound} below.

\begin{lemma}\label{lemma:coarsening-monotonicity}
If $q \in \N$ and $\Psi$ is an instance of $\mPi$, then $\ocspval_\Psi \geq \cspval_{\Psiq}$.
\end{lemma}

\begin{proof}
We will show that for every assignment $\vecb \in [q]^n$ to $\Psiq$, we can construct an assignment $\vecb^{\refn} \in \sym_n$ to $\Psi$ such that $\ocspval_\Psi(\vecb^{\refn}) \geq \cspval_{\Psiq}(\vecb)$. Consider an assignment $\vecb \in [q]^n$. Let $\vecb^{\refn}$ be the ordering on $[n]$ constructed by placing the blocks $\vecb^{-1}(1),\ldots,\vecb^{-1}(q)$ in order (within each block, we enumerate the indices arbitrarily). Consider any constraint $C = \vecj = (j_1,\ldots,j_k)$ in $\Psi$ which is satisfied by $\vecb$ in $\Psiq$. Since $\Piq(\vecb|_{\vecj}) = 1$, by definition of $\Piq$ we have that $\Pi(\ord(\vecb|_\vecj)) =1$ and in particular that $b_{j_1},\ldots,b_{j_k}$ are distinct. The latter implies, by the construction of $\vecb^{\refn}$, that $\ord(\vecb|_\vecj) = \ord(\vecb^{\refn}|_\vecj)$. Hence $\Pi(\ord(\vecb^{\refn}|_\vecj)) = 1$, so $\vecb^{\refn}$ satisfies $C$ in $\Psi$. Hence $\ocspval_\Psi(\vecb^{\refn}) \geq \cspval_{\Psiq}(\vecb)$.
\end{proof}

\subsection{$\cY$ has high $\mPi$ values}\label{sec:y-lower-bound}

In this section, we prove \ref{lemma:y-lower-bound}, which states that the $\mPi$ values of instances $\Psi$ drawn from $\DY$ are large. Note that we prove a bound for \emph{every} instance $\Psi$ in the support of $\DY$, although it would suffice for our application to prove that such a bound holds with high probability over the choice of $\Psi$.

To prove \ref{lemma:y-lower-bound}, by \ref{lemma:coarsening-monotonicity}, it suffices to show that $\cspval_{\Psiq} \geq 1-\frac{k-1}q$. One natural approach is to consider the $q$-partition $\vecb = (b_1,\ldots,b_n) \in [q]^n$ sampled when sampling $\Psi$ and view $\vecb$ as an assignment for $\Psiq$. Consider any constraint $C = \vecj = (j_1,\ldots,j_k)$ in $\Psi$; by the definition of $\cY^{\Pi,\vecpi}$ (\ref{def:YES_NO_MaxOCSP_dist}), we have $\vecb|_{\vecj} = \vecv^{(c)} +_q \vecpi$ for some (unique) $c \in [q]$, where $\vecv^{(c)}$ denotes the vector $(c,\ldots,c) \in [q]^k$. We term $c$ the \emph{identifier} of $C$. Now we use the following simple fact:

\begin{fact}\label{fact:shifted-order}
Let $\vecpi \in \sym_k$. Then for every $c \in \{1,\ldots,q-k\}\cup\{q\}$, $\ord(\vecv^{(c)}+_q\vecpi) = \vecpi$.
\end{fact}

\begin{proof}
Follows from the fact that for $c$ in this range, and every $i,j \in [k] \subset [q]$, we have $i<j$ iff $i+_qc < j+_qc$.
\end{proof}

Thus, $C$ is satisfied by $\vecb$ iff $\Pi(\ord(\vecv^{(c)}+_q \vecpi)) = 1$.  Hence a sufficient condition for $\vecb$ to satisfy $C$ is that $C$'s identifier $c \in \{1,\ldots,q-k\}\cup\{q\}$, since in this case by \ref{fact:shifted-order}, we have $\ord(\vecv^{(c)} +_q \vecpi) = \vecpi$. Unfortunately, when sampling the constraints, we might get ``unlucky'' and get a sample that over-represents the identifiers $ \{q-k+1,\ldots,q-1\}$. We can resolve this issue using ``shifted'' versions of $\vecb$.\footnote{Alternatively, in expectation, $\cspval_{\Psiq}(\vecb) = 1-\frac{k-1}q$. Hence with probability at least $\frac{99}{100}$, $\cspval_{\Psiq}(\vecb) \geq 1-\frac{100(k-1)}q$ by Markov's inequality; this suffices for a ``with-high-probability'' statement.} The proof is as follows:

\begin{proof}[Proof of \ref{lemma:y-lower-bound}]
For $t \in [q]$, define the assignment $\vecb^{(t)} = (b^{(t)}_1,\ldots,b^{(t)}_n)$ to $\Psiq$ via $b^{(t)}_i = b_i+_qt$ for $i \in [n]$.

Fix $t \in [q]$. Then we claim that $\vecb^{(t)}$ satisfies any constraint $C$ with identifier $c$ such that $c+_qt \in \{1,\ldots,q-k\}\cup\{q\}$. Indeed, if $C = \vecj$ is a constraint with identifier $c$, we have $\vecb^{(t)}|_{\vecj} = \vecv^{(c)}+_q\vecv^{(t)}+_q\vecpi=\vecv^{(c+_qt)} +_q \vecpi$, and then we use \ref{fact:shifted-order}.

Now (no longer fixing $t$), for each $c \in [q]$, let $w^{(c)}$ be the fraction of constraints in $\Psi$ with identifier $c$. By the previous paragraph, for each $t \in [q]$, we have $\cspval_\Psiq(\vecb^{(t)}) \geq \sum_{c:c+_qt \in \{1,\ldots,q-k\}\cup\{q\}} w^{(c)}$. On the other hand, $\sum_{c=1}^q w^{(c)} = 1$ (since every constraint has some (unique) identifier). Hence \[ \sum_{t=1}^q \cspval_{\Psiq}(\vecb^{(t)}) \geq \sum_{t=1}^q \left(\sum_{c : c+_qt \in \{1,\ldots,q-k\}\cup\{q\}} w^{(c)} \right) = q-(k-1), \] since each term $w^{(c)}$ appears exactly $q-(k-1)$ times in the expanded sum. Hence by averaging, there exists some $t \in [q]$ such that $\cspval_{\Psiq}(\vecb^{(t)}) \geq 1-\frac{k-1}q$, and so $\cspval_{\Psiq} \geq 1-\frac{k-1}q$, as desired.
\end{proof}

\subsection{$\cN$ has low $\mPi$ values}\label{sec:n-upper-bound}

In this section, we prove \ref{lemma:n-upper-bound}, which states that the $\mPi$ value of an instance drawn from $\cN$ does not significantly exceed the random ordering threshold $\rho(\Pi)$, with high probability.

\begin{remark}
    Using concentration bounds (i.e., \ref{lemma:azuma}), one could show that the probability that a fixed solution $\vecsigma \in \sym_n$ satisfies more than $\rho(\Pi) + 1/q$ constraints is exponentially small in $n$. However, taking a union bound over all $n!$ permutations $\vecsigma$ would cause an unacceptable blowup in the probability (since by Stirling's approximation, $n! \sim (n/e)^n$).
\end{remark}

Instead, to prove \ref{lemma:n-upper-bound}, we take an indirect approach, involving bounding the $\mcsp$ value of the $q$-coarsening of a random instance and bounding the gap between the $\mocsp$ value and the $q$-coarsened $\mcsp$ value. To do this, we define the following notions of small set expansion for $k$-hypergraphs:

\begin{definition}[Lying on a set]
Let $G = ([n],E)$ be a $k$-hypergraph. Given a set $S \subseteq [n]$, a $k$-hyperedge $\vecj \in E$ \emph{lies on} $S$ if it is incident on two (distinct) vertices in $S$ (i.e., if $|\Gamma(\vecj) \cap S| \geq 2$).
\end{definition}

\begin{definition}[Congregating on a partition]
Let $G = ([n],E)$ be a $k$-hypergraph. Given a $q$-partition $\vecb \in [q]^n$, a $k$-hyperedge $\vecj \in E$ \emph{congregates} on $\vecb$ if it lies on a block $\vecb^{-1}(i) = \{j \in [n]:b_j=i\}$ for some $i \in [q]$.
\end{definition}

\begin{definition}[Small set hypergraph expansion (SSHE) property]
A $k$-hypergraph $G = ([n],E)$ is a \emph{$(\gamma,\delta)$-small set hypergraph expander (SSHE)} if it has the following property: For every subset $S \subseteq [n]$ of size at most $\gamma n$, the number of $k$-hyperedges in $E$ which lie on $S$ is at most $\delta |E|$.
\end{definition}

\begin{definition}[Small partition hypergraph expansion (SPHE) property]\label{def:sphe}
A $k$-hypergraph $G = ([n],E)$ is a \emph{$(\gamma,\delta)$-small partition hypergraph expander (SPHE)} if it has the following property: For every partition $\vecb \in [q]^n$ where each block $\vecb^{-1}(i) = \{j \in [n]:b_j=i\}$ has size at most $\gamma n$, the number of $k$-hyperedges in $E$ that congregate on $\vecb$ is at most $\delta |E|$.
\end{definition}

In the context of \ref{fig:mas}, the SPHE property says that for \emph{any} partition with small blocks, there cannot be too many ``orange'' edges.

In the remainder of this subsection, we state several lemmas and then give a formal proof of \ref{lemma:n-upper-bound}. We begin with several short lemmas.

\begin{lemma}[Good SSHEs are good SPHEs]\label{lemma:sshe-to-sphe}
For every $k \in \N$ and $\gamma,\delta > 0$, if a $k$-hypergraph $G = (V,E)$ is a $(\gamma,\delta)$-SSHE, then it is a $(\gamma,\delta (2/\gamma+1))$-SPHE.
\end{lemma}

\begin{proof}
Let $n = |V|$. Consider any partition $\vecb \in [q]^n$ of $V$ where each block $\vecb^{-1}(i)$ has size at most $\gamma n$. WLOG, all but one block has size at least $\frac{\gamma n}2$ (if not, merge blocks until this happens, only increasing the number of $k$-hyperedges that congregate on $\vecb$). Hence $\ell \leq \frac2\gamma+1$.\footnote{We include the $+1$ to account for the extra block which may have arbitrarily small size. Excluding this block, there are at most $\frac{n}{\lceil \gamma n /2\rceil} \leq \frac{n}{\gamma n / 2}$ blocks remaining.} By the SSHE property, there are at most $\delta m$ $k$-hyperedges that lie on each block; hence there are at most $\delta (\frac2\gamma+1) m$ constraints that congregate on $\vecb$.
\end{proof}

\begin{lemma}[Coarsening roughly preserves value in SPHEs]\label{lemma:sphe-gap-bound}
Let $\Psi$ be a $\mPi$ instance on $n$ variables and let $q \geq 2/\gamma$. Suppose that the constraint hypergraph $G(\Psi)$ of $\Psi$ is a $(\gamma,\delta)$-SPHE. Then for sufficiently large $n$, \[ \ocspval_\Psi \leq \cspval_{\Psiq} + \delta. \]
\end{lemma}

\begin{proof}
For every assignment $\vecsigma=(\sigma_1,\ldots,\sigma_n) \in \sym_n$ to $\Psi$, we will construct an assignment $\vecsigma^{\qcoar} = (\sigma^{\qcoar}_1,\ldots,\sigma^{\qcoar}_n) \in [q]^n$ to $\Psiq$ such that $\ocspval_\Psi(\vecsigma) \leq \cspval_{\Psiq}(\vecsigma^{\qcoar}) + \delta$. Fix $\vecsigma \in \sym_n$. Define $\vecsigma^{\qcoar} \in [q]^n$ by $\sigma^{\qcoar}_i = \lfloor \sigma_i / \gamma n \rfloor$ for each $i \in [n]$. We verify that since $\sigma_i \leq n$, we have \[ \sigma^{\qcoar}_i \leq \lfloor n / \gamma n \rfloor \leq n/\gamma n \leq n/(2n/q) \leq n/(n/q) = q, \] so 
$\vecsigma^{\qcoar}$ is a valid assignment to $\Psiq$. Also, $\vecsigma^{\qcoar}$ has the property that for every $i,j \in [n]$, if $\sigma_i < \sigma_j$ then $\sigma^{\qcoar}_i \leq \sigma^{\qcoar}_j$; we call this \emph{monotonicity} of $\vecsigma^{\qcoar}$.

Now view $\vecsigma^{\qcoar}$ as a $q$-partition and consider the constraint hypergraph $G(\Psi)$ of $\Psi$ (which is the same as the constraint hypergraph $G(\Psiq)$ of $\Psiq$). Call a constraint $C = \vecj = (j_1,\ldots,j_k)$ in $\Psi$ \emph{good} if it is both satisfied by $\vecsigma$, \emph{and} the $k$-hyperedge corresponding to it does not congregate on $\vecsigma^{\qcoar}$. If $C$ is good, then $\sigma^{\qcoar}_{j_1},\ldots,\sigma^{\qcoar}_{j_k}$ are all distinct; together with monotonicity of $\vecsigma^{\qcoar}$, we conclude that if $C$ is good, then $\ord(\vecsigma^{\qcoar}|_\vecj) = \ord(\vecsigma|_\vecj)$, and hence $C$ is also satisfied by $\vecsigma^{\qcoar}$ in $\Psiq$.

Finally, we note that each block in $\vecsigma^{\qcoar}$ has size at most $\gamma n$ by definition; hence by the SPHE property of the constraint hypergraph of $\Psi$, at most $\delta$-fraction of the constraints of $\Psi$ correspond to $k$-hyperedges that congregate on $\vecsigma^{\qcoar}$. Since $\ocspval_\Psi(\vecsigma)$-fraction of the constraints of $\Psi$ are satisfied by $\vecsigma$, at least $(\ocspval_\Psi(\vecsigma) - \delta)$-fraction of the constraints of $\Psi$ are good, and hence $\vecsigma^{\qcoar}$ satisfies at least $(\ocspval_\Psi(\vecsigma) - \delta)$-fraction of the constraints of $\Psiq$, as desired.
\end{proof}

The construction in this lemma was called \emph{coarsening} the assignment $\vecsigma$ by~\textcite{GHM+11} (cf.~\cite[Definition 4.1]{GHM+11}).

We also include the following helpful lemma, which lets us restrict to the case where our sampled $\mPi$ instance has many constraints.

\begin{lemma}[Most instances in $\cN$ have many constraints]\label{lemma:n-edge-count}
For every $n$, $\alpha,\gamma > 0$, and $q \in \N$, \[ \Pr_{\Psi \sim \DN}\left[m(\Psi) \leq \left(\frac{\alpha T}{2q^k}\right) n \right] \leq \exp\left(-\left(\frac{\alpha T}{8q^k} \right)n\right). \]
\end{lemma}

\begin{proof}
The number of constraints in $\Psi$ is distributed as the sum of $n\alpha T$ independent Bernoulli$(1/q^k)$ random variables. The desired bound follows by applying a Chernoff bound (\ref{lemma:chernoff}) with $\eta = p/2$.
\end{proof}

Now we state the following pair of lemmas, whose more involved proofs we defer to \ref{sec:n-sshe} and \ref{sec:n-coarse-lb}, respectively:

\begin{lemma}\label{lemma:n-sshe}
For every $n$, $\alpha, \gamma > 0$, and $q \in \N$ with $\alpha \leq 1/(2k)$,
\begin{multline*}
    \Pr_{\Psi \sim \DN}\left[G(\Psi)\text{ is not a }(\gamma,8k^2\gamma^2)\text{-SSHE} \bigbar m(\Psi) \geq \frac{n\alpha T}{2q^k} \right] \\
    \leq \exp\left(-{\left(\frac{\gamma^2 \alpha T}{2q^k}-\ln 2\right)}n\right).
\end{multline*}
\end{lemma}

\begin{lemma}[Satisfiability of random $\mcsp(\Psiq)$ instances]\label{lemma:n-coarse-lb}
For every $n$, $\alpha, \eta >0$,
\begin{multline*}
\Pr_{\Psi \sim \DN} \left[\cspval_{\Psiq} \geq \rho(\Pi)+\eta \bigbar m(\Psi) \geq \frac{n\alpha T}{2q^k} \right] \\
\leq \exp\left(-{\left(\frac{\eta^2\alpha T}{4(\rho(\Pi)+\eta) q^k}-\ln q\right)}n\right).
\end{multline*}
\end{lemma}

We remark here that the proofs of both lemmas only require union bounds over sets of size $(O_\epsilon(1))^n$ (the set of all small subsets of $[n]$ and of all solutions to the coarsened $\mcsp$, respectively); this lets us avoid the issue, mentioned in the Remark at the beginning of this subsection, of union-bounding over the entire space $\sym_n$ of super-exponential size $n!$ directly.

We finally give the proof of \ref{lemma:n-upper-bound}.

\begin{proof}[Proof of \ref{lemma:n-upper-bound}]
Let $q_0 \eqdef \left\lceil \frac{192k^2}{\epsilon} \right\rceil$ and let $\alpha_0 \eqdef \frac1{2k}$. Suppose $\alpha \leq \alpha_0$ and $q \geq q_0$. Then let $\gamma \eqdef \frac{\epsilon}{96k^2}$ and $\eta \eqdef \frac{\epsilon}4$, and let \[ T_0 \eqdef \left\lceil \max\left\{\frac{2(\ln 2)q^k}{\gamma^2 \alpha},\frac{4(\rho(\Pi)+\eta)q^k (\ln q)}{\eta^2 \alpha}\right\} \right\rceil + 1. \]
Now, we will prove the desired bound for any $T \geq T_0$.

Let $\cE_1,\cE_2,$ and $\cE_3$ denote, respectively, the events ``$m(\Psi) \leq (\alpha T/(2q^k)) n$'', ``$G(\Psi)$ is not a $(\gamma,8k^2\gamma^2)$-SSHE'', and ``$\cspval_{\Psi^\qcoar} \geq \rho(\Pi) + \eta$''. Then, since $\alpha \leq \alpha_0$, \ref{lemma:n-edge-count}, \ref{lemma:n-sshe}, and \ref{lemma:n-coarse-lb} state that
\begin{align*}
    \Pr[\cE_1] &\leq \exp(-C_1 n), \\
    \Pr[\cE_2 \mid \overline{\cE_1}] &\leq \exp(-C_2 n),\\
    \Pr[\cE_3 \mid \overline{\cE_1}] &\leq \exp(-C_3 n),
\end{align*}
respectively, where we define the constants $C_1 = \frac{\alpha T}{8q^k}$, $C_2 = \frac{\gamma^2 \alpha T}{2q^k} - \ln2$, and $C_3 = \frac{\eta^2 \alpha T}{4(\rho(\Pi) + \eta)q^k} - \ln q$, and all probabilities are over the choice of $\Psi \sim \DN$. Now observe that for our choice of $T$, $C_1,C_2,C_3$ are all positive and do not depend on $n$, so for sufficiently large $n$, all three probabilities are smaller than $1/1000$. This implies that $\Pr[\cE_2 \vee \cE_3] \leq 1/100$: Indeed,
\begin{align*}
    \Pr[\cE_2 \vee \cE_3] &= \Pr[\cE_2 \vee \cE_3 \mid \overline{\cE_1}] \Pr[\overline{\cE_1}] + \Pr[\cE_2 \vee \cE_3 \mid \cE_1] \Pr[\cE_1] \tag{total probability} \\
    &\leq \Pr[\cE_2 \vee \cE_3 \mid \overline{\cE_1}] + \Pr[\cE_1] \tag{probabilities $\leq 1$} \\
    &\leq \Pr[\cE_2 \mid \overline{\cE_1}] + \Pr[\cE_3 \mid \overline{\cE_1}] + \Pr[\cE_1] \tag{union bound}
\end{align*}
and this is $\leq 3/1000 < 1/100$ by assumption.

Finally, we show that when neither $\cE_2$ nor $\cE_3$ occurs, we have $\ocspval_\Psi \leq \rho(\Pi) + \epsilon/2$. Indeed, if $G(\Psi)$ is a $(\gamma,8k^2\gamma^2)$-SSHE, by \ref{lemma:sshe-to-sphe} it is also a $(\gamma,\delta)$-SPHE for $\delta = 8k^2 \gamma^2 (2/\gamma+1) \leq 8k^2 \gamma^2 (3/\gamma) = 24k^2 \gamma = \epsilon/4$. Now since $q \geq q_0 \geq 2/\gamma$, we can apply \ref{lemma:sphe-gap-bound} and conclude that \[ \ocspval_\Psi \leq \rho(\Pi) + \eta + \delta \leq \rho(\Pi) + \frac{\epsilon}2, \] as desired.
\end{proof}

\subsubsection{$\cN$ is a good SSHE with high probability: Proving \ref{lemma:n-sshe}}\label{sec:n-sshe}

Recall that for a $k$-hypergraph $G = (V,E)$ and $S\subseteq V(G)$, we define $\Gamma(\vecj) \subseteq V$ as the set of vertices incident on $\vecj$.

\begin{proof}[Proof of \ref{lemma:n-sshe}]
Let $\tilde{G}_t$ denote the $t$-th hypermatching sampled when sampling $\Psi$ (as in \ref{def:YES_NO_MaxOCSP_dist}). Let $\beta = \frac{m(\Psi)}{Tn}$, and for each $t \in [T]$, let $\beta_t = \frac{m(G_t)}n$, so that $\beta = \frac1T\sum_{t=1}^T \beta_t$. By our conditioning assumption, $\beta \geq \frac{\alpha}{2q^k}$, and $\beta_t \leq \alpha$ for each $t \in [T]$. It suffices to prove the lemma conditioned on fixed values for $\beta_1,\ldots,\beta_T$. This is equivalent to simply sampling hypermatchings $G_t \sim \matchings_{k,\beta_t}(n)$ independently and including all of their $k$-hyperedges as constraints.

Fix any set $S \subseteq [n]$ of size at most $\gamma n$. For each $t \in [T]$, label the $k$-hyperedges of $G_t$ as $\vecj(t,1),\ldots,\vecj(t,\beta_t n)$. Consider the collection of $m(\Psi) = \beta T n$ random variables $\{X_{t,i}\}_{t\in[T],i\in[\beta_tn]}$, each of which is the indicator for the event that $\vecj(t,i)$ lies on $S$. Define $X = \sum_{t=1}^T\sum_{i=1}^{\beta_tn} X_{t,i}$ as the number of such ``lying on'' edges we observe. Since $m(G(\Psi)) = \beta T n$, it suffices to bound the probability that $X \geq 8k^2\gamma^2 m$, and then take the union bound over all subsets $S$.

For fixed $t \in [T]$, we first bound $\Exp[X_{t,i} \mid X_{t,1},\ldots,X_{t,i-1}]$ for each $i\in[\beta_tn]$. Conditioned on $\vecj_{t,1},\ldots,\vecj_{t,i-1}$, the $k$-hyperedge $\vecj(t,i)$ is uniformly distributed over the set of all $k$-hyperedges on $[n] \setminus (\Gamma(\vecj_{t,1}) \cup \cdots \cup \Gamma(\vecj_{t,i-1}))$. It suffices to union-bound, over distinct pairs $\{a,b\} \in \binom{[k]}2$, the probability that the $a$-th and $b$-th vertices of $\vecj(t,i)$ are in $S$ (conditioned on $X_{t,1},\ldots,X_{t,i-1}$). We can sample the $a$-th and $b$-th vertices of $\vecj(t,i)$ first (uniformly over the remaining vertices) and then ignore the remaining vertices. Hence we have the upper bound 

\begin{align*}
    \Exp[X_{t,i} \mid X_{t,1},\ldots,X_{t,i-1}] & \leq \binom{k}2 \cdot \frac{|S|(|S|-1)}{(n-k(i-1))(n-k(i-1)-1)}\\
    &\leq \binom{k}2 \cdot \left(\frac{|S|}{n-k(i-1)}\right)^2 \\
    &\leq \binom{k}2 \cdot \left(\frac{|S|}{n-k\beta_t n}\right)^2 \leq 4k^2\gamma^2 \, , 
\end{align*} since $\beta_t \leq \alpha \leq \frac1{2k}$.

Now, we want to apply \ref{lemma:azuma} to the random variables $(X_{t,i})$. Consider enumerating the $X_{t,i}$'s lexicographically as \[ X_{1,1},\ldots,X_{1,\beta_1 n},X_{2,1},\ldots,X_{2,\beta_2n},\ldots,X_{T,1},\ldots,X_{T,\beta_Tn}. \] Now since the hypermatchings $G_t$ are sampled independently for all $t \in [T]$, the collections of random variables $(X_{t,i})_{i \in [\beta_t n]}$ are independent across $t \in [T]$. So, for all $t \in [T]$ and $i \in [\beta_t n]$, the expectation of $X_{t,i}$ conditioned on the lexicographically-earlier $X_{t',i'}$'s is still at most $4k^2\gamma^2$. Thus \ref{lemma:azuma} (at $p = \eta = 4k^2 \gamma^2$) implies that \[ \Pr_{\Psi \sim \DN}\left[X \geq 8 k^2 \gamma^2 m \right] \leq \exp\left(-k^2 \gamma^2 m \right) \leq \exp(-\gamma^2 m). \] Finally, we assumed $m \geq \frac{\alpha T}{2q^k} n$, combined with the union-bound over the $\leq 2^n$ possible subsets $S \subseteq [n]$, gives the desired bound.
\end{proof}

\subsubsection{$\cN$ has low coarsened $\mcsp(\Psiq)$ values (w.h.p.): Proving \ref{lemma:n-coarse-lb}}\label{sec:n-coarse-lb}

\begin{proof}[Proof of \ref{lemma:n-coarse-lb}]
We consider the same setup as in the first paragraph of the proof of \ref{lemma:n-sshe} in the previous subsection: We set $\beta = \frac{m(\Psi)}{Tn}$ and $\beta_t = \frac{m(G_t)}{n}$. Abbreviating $m=m(\Psi)$, $\Psi$ contains $m = \beta T n$ constraints. We condition on fixed $\beta_1,\ldots,\beta_T$, and consider sampling hypermatchings $G_t \sim \matchings_{k,\beta_t}(n)$ and we have $\frac1T \sum_{t=1}^T \beta_t = \beta \geq \alpha/2q^k$. We label the $k$-hyperedges of $G_t$ as $\vecj(t,1),\ldots,\vecj(t,\beta_t n)$. Now, we view these $k$-hyperedges as the constraints of $\Psiq$ for $\Psi \sim \DN$. 

Let $X_{t,i}$ for $t \in [T],i \in [\beta_t n]$ be the indicator for the event that, as a constraint in $\Psiq$, $\vecj(t,i)$ is satisfied by $\vecb$, i.e., $\Piq(\vecb|_{\vecj(t,i)}) = 1$. Again, like in the proof of \ref{lemma:n-sshe}, we first fix $t \in [T]$ and bound $\Exp[X_{t,i} \mid X_{t,1},\ldots,X_{t,i-1}]$, for each $i \in [\beta_t n]$. Conditioned on $\vecj(t,1),\ldots,\vecj(t,i-1)$, the $k$-hyperedge $\vecj(t,i)$ is uniformly distributed over the set of all $k$-hyperedges on $[n] \setminus (\Gamma(\vecj(t,1)) \cup \cdots \cup \Gamma(\vecj(t,i-1)))$. Now, we claim that $\Exp[X_{t,i} \mid X_{t,1},\ldots,X_{t,i-1}] \leq \rho(\Pi)$. Indeed, the set of possible $k$-hyperedges on $[n] \setminus (\Gamma(\vecj(t,1)) \cup \cdots \cup \Gamma(\vecj(t,i-1)))$ may be partitioned into blocks of size $k!$ by mapping each $k$-hyperedge to the set of vertices on which it is incident. That is, we can consider the $k!$ possible $k$-hyperedges resulting from permuting the vertices of a given $k$-hyperedge $\vecj = (j_1,\ldots,j_k)$. If $b_{j_1},\ldots,b_{j_k}$ are not all distinct, then $\ord(\vecb|_\vecj) = \bot$ and so none of the $k$-hyperedges are satisfied by $\vecb$ as constraints in $\Psiq$; otherwise, $\ord(\vecb|_\vecj)$ is some ordering in $\sym_k$, and so exactly $|\supp(\Pi)| = \rho(\Pi) \cdot k!$ permutations of $\vecj$ are satisfied as constraints in $\Psiq$.

Now, we again apply \ref{lemma:azuma} to the random variables $(X_{t,i})$. Using the same lexicographic enumeration of variables as in the proof of \ref{lemma:n-sshe}, and the same observation that the hypermatchings $G_t$ are sampled independently, we conclude that the expectation of $X_{t,i}$ conditioned on the lexicographically-earlier $X_{t',i'}$'s is at most $\rho(\Pi)$. Again, we apply \ref{lemma:azuma} (now with $p=\rho(\Pi)$) to deduce that \[ \Pr\left[\sum_{t=1}^T \sum_{i=1}^{\beta_t n} X_i \geq (\rho(\Pi) + \eta) m\right] \leq \exp\left(- \left(\frac{\eta^2}{2(\rho(\Pi) + \eta)}\right) m \right). \] Finally, we assumed $m \geq n\alpha T/2q^k$, so the RHS is at most \[ \exp\left(- \left(\frac{\eta^2 \alpha T}{4(\rho(\Pi) + \eta)}\right) n \right); \] then, we union bound over $\vecb \in [q]^n$ to yield the desired conclusion.

\end{proof}

\section{Streaming indistinguishability of $\cY$ and $\cN$}\label{sec:streaming}

In this section, we prove \ref{lem:our-indist} establishing the streaming indistinguishability of the distributions $\cY$ and $\cN$. This indistinguishability follows directly from the work of \textcite{CGS+22-linear-space}, who introduce a $T$-player communication problem called \emph{implicit randomized mask detection (IRMD)}. Once we properly situate our instances $\cY$ and $\cN$ within the framework of \cite{CGS+22-linear-space}, \ref{lem:our-indist} follows immediately.

We first recall their definition of the IRMD problem and state their lower bound. The following definition is based on \cite[Definition 3.1]{CGS+22-linear-space}. In \cite{CGS+22-linear-space}, the IRMD game is parametrized by two distributions $\cD_Y$ and $\cD_N$, but hardness is proved for a specific pair of distributions which suffices for our purpose; these distributions will thus be ``hardcoded'' into the definition we give.

\begin{definition}[Implicit randomized mask detection problem\label{def:irmd}]
Let $q,k,n,T \in \mathbb{N},\alpha \in (0,1/k)$ be parameters. $\IRMD_{q,k,\alpha,T}(n)$ is a $T$-player one-way communication game. The $T$ players are denoted $\Player_1,\ldots,\Player_T$. The input to the players is drawn either from a ``$\yes$ distribution'' or a ``$\no$ distribution'', and their collective goal is to distinguish these two cases. Each player $\Player_t$ receives two inputs: (i) a uniform $k$-hypermatching $\tilde{G}_t \sim \matchings_{k,\alpha}(n)$ on $n$ vertices with $\alpha n$ hyperedges $\vecj(t,1),\ldots,\vecj(t,\alpha n)$, and (ii) a vector $\vecz_t = (\vecz_{t,1},\ldots,\vecz_{t,\alpha n}) \in ([q]^k)^{\alpha n}$ labeling the $k$-hyperedges of $\tilde{G}_t$. There is also a random partition $\vecb \sim \unif([q]^n)$ which is hidden, i.e., not given to the players directly. The difference between the $\yes$ and $\no$ distributions is in how $\vecz_t$ is determined by $\tilde{G}_t$ and $\vecb$:

\begin{itemize}
    \item $\yes$ case: Each $\vecz_{t,i} = \vecb|_{\vecj(t,i)} +_q \vecy_{t,i}$ where $\vecy_{t,i} \sim \unif(\{\vecv^{(c)}=(c,\ldots,c) : c \in [q]\})$ independently.
    \item $\no$ case: Each $\vecz_{t,i} = \vecb|_{\vecj(t,i)} +_q \vecy_{t,i}$ where $\vecy_{t,i} \sim \unif([q]^k)$ independently.
\end{itemize}

$\Player_t$ sends a private message to the $\Player_{t+1}$ after receiving a message from $\Player_{t-1}$. The goal is for $\Player_T$ to decide whether $(\vecz_t)_{t \in T}$ has been chosen from the $\yes$ or $\no$ distribution, and the advantage of a protocol is defined as \[ \left\lvert\Pr_{\yes\text{ case}}[\Player_T\text{ outputs 1}]-\Pr_{\no\text{ case}}[\Player_T\text{ outputs 1}]\right\rvert. \]
\end{definition}

Note that the definition of the IRMD problem does not depend on an underlying predicate (beyond fixing $q$ and $k$). Nevertheless, we will be able to leverage IRMD's hardness to prove \ref{lem:our-indist} (and indeed, all hardness results in \cite{CGS+22-linear-space} itself stem from hardness for the IRMD problem). The following theorem of \cite{CGS+22-linear-space} gives a lower bound on the communication complexity of the IRMD problem:

\begin{theorem}[{\cite[Theorem~3.2]{CGS+22-linear-space}}]\label{thm:distinguishing_distributions}
For every $q,k \in \mathbb{N}$ and $\delta \in (0,1/2)$, $\alpha \in (0,1/k)$, $T \in \mathbb{N}$ there exists $n_0 \in \mathbb{N}$ and $\tau \in (0,1)$ such that the following holds. For all $n \geq n_0$, every protocol for $\IRMD_{q,k,\alpha,T}(n)$ with advantage $\delta$ requires $\tau n$ bits of communication.
\end{theorem}

Now, we use this hardness result to prove \ref{lem:our-indist}. The following proof is based on the proof of \cite[Theorem 4.3]{CGS+22-linear-space}. However, \cite[Theorem 4.3]{CGS+22-linear-space} as written contains both the streaming-to-communication reduction and an analysis of the CSP values of $\yes$ and $\no$ instances; in the following, we reprove only the former (and adapt the language to our setting).

\begin{proof}[Proof of \ref{lem:our-indist}]
We prove the lemma for the same $\alpha_0$ as in \ref{thm:distinguishing_distributions}.

Suppose $\Alg$ is a $s(n)$-space streaming algorithm that distinguishes $\DY$ from $\DN$ with advantage $1/8$ for all lengths $n$. We now show how to use $\Alg$ to construct a protocol $\Prot = (\Prot_1,\ldots,\Prot_T)$ (where $\Prot_t$ defines the behavior of the $\Player_t$) solving $\IRMD_{q,k,\alpha,T}(n)$ with advantage $1/8$ for $n \geq n_0$, which uses only $s(n)$ bits of communication; \ref{thm:distinguishing_distributions} provides a constant $\tau \in (0,1)$ yielding the desired contradiction if we set $s(n) \leq \tau n$. As is standard, this reduction will involve the players collectively running the streaming algorithm $\Alg$. That is, $\Prot$ is defined as follows: For $t \in [T-2]$, $\Prot_t$ adds some constraints to the stream (in a manner to be specified below) and then sends the state of $\Alg$ on to $\Player_{t+1}$. Finally, $\Prot_T$ terminates the streaming algorithm and outputs the output of $\Alg$.

Recall, $\Player_t$'s input is a pair $(\tilde{G}_t, \vecz_t)$ consisting of a hypermatching $\tilde{G}_t$ and a vector $\vecz_t = (\vecz_{t,1},\ldots,\vecz_{t,\alpha n})$ of labels $\vecz_{t,i} \in [q]^k$ for each hyperedge in $M_t$. We define $\Prot_t$'s behavior as follows: It adds each hyperedge $\vecj(i)$ in $\tilde{G}_t$ to the stream iff $\vecz_{t,i} = \vecpi$.

Let $\cY'(n)$ and $\cN'(n)$ denote the distributions of $\mPi$ instances constructed by $\Prot$ in the $\yes$ and $\no$ cases, respectively. The crucial claim is that $\cY'(n)$ and $\DY$ are identical distributions, and similarly with $\cN'(n)$ and $\DN$. This claim suffices to prove the lemma since the constructed stream of constraints is fed into $\Alg$, which is an $s(n)$-space streaming algorithm distinguishing $\DY$ from $\DN$; hence we can conclude that $\Prot$ is a protocol for $\IRMD$ using $s(n) \leq \tau n$ bits of communication.

It remains to prove the claim. We first consider the $\no$ case. $\cN'(n)$ and $\DN$ are both sampled by independently sampling $T$ hypermatchings from $\matchings_{k,\alpha}(n)$ and then (independently) selecting some subset of $k$-hyperedges from each hypermatching to add as constraints. It suffices by independence to prove equivalence of how the subset of each hypermatching is sampled in each case. For each $t \in [T]$, $\Prot_t$ adds a hyperedge $\vecj(t,i)$ iff $\vecz_{t,i} = \vecpi$. But in the $\no$ case (even conditioned on all other $\vecz_{t,i'}$'s, on the hidden partition $\vecb$, and on $\vecj(t,i)$ itself), $\vecz_{t,i}$ is a uniform value in $[q]^k$, and hence $\vecj(t,i)$ is added to the instance with probability $1/q^k$. This is exactly how we defined $\DN$ to sample constraints.

Similarly, in the $\yes$ case, we consider the sampled $q$-partition $\vecb = (b_1,\ldots,b_n) \in [q]^n$ and a hyperedge $\vecj(t,i) = (j_1,\ldots,j_k)$. In this case, by the definition of $\IRMD$, we have $\vecz_{t,i} = \vecb|_{\vecj(t,i)} +_q \vecv^{(c)}$ where $c \sim [q]$ is uniform and $\vecv^{(c)} = (c,\ldots,c)$. Hence $\Prot_t$ will add $\vecj(t,i)$ iff $\vecb|_{\vecj(t,i)} = \vecpi +_q \vecv^{(c')}$ where $c' \in [q]$ is the unique value such that $c' +_q c = q$. Consider the event $\cE_1$ that there exists any $c' \in [q]$ such that $\vecb|_{\vecj(t,i)} = \vecpi +_q \vecv^{(c')}$ and the event $\cE_2$ that $c+_qc'=q$. Note that if $\cE_1$ does not occur, then $\vecj(t,i)$ is never added, while if $\cE_1$ does occur, with probability $1/q$ over the choice of $c$, $\cE_2$ occurs and hence $\vecj(t,i)$ is added. (Again, this holds even conditioned on all other $\vecz_{t,i'}$'s, on $\vecj(t,i)$, and on $\vecb$.) This is exactly how we defined $\DY$ to sample constraints.
\end{proof}

\printbibliography

\appendix
\section{Good approximations in $\tilde{O}(n)$ space via sparsification}\label{app:sparse-alg}

In this appendix, we prove \ref{thm:sparse-alg}, which gives a simple $(1-\epsilon)$-approximation algorithm for all $\mocsp$ problems given $\tilde{O}(n)$ space. To prove \ref{thm:sparse-alg}, we first develop a simple sampling lemma that states that sparsifying down to $\Omega(n \log n/\epsilon^2)$ constraints preserves the value of a $\mocsp$ instance up to an additive $\pm \epsilon$:

\begin{lemma}[$\tilde{O}(n)$-space sampling lemma]\label{lem:sparse}
    Let $\cF = \bigcup_{k \in \N}\{\Pi : \sym_k \to \{0,1\}\}$ denote the (infinite) family of all ordering predicates. For all $\epsilon > 0$ and sufficiently large $n$, let $\Psi$ be an instance of $\mocsp(\cF)$ on $n$ variables. Let $\tilde{m} \geq 10 n \log n /\epsilon^2$. Consider the random instance $\tilde{\Psi}$ with $n$ variables and $\tilde{m}$ constraints which is sampled by sampling each constraint independently and uniformly from the constraints of $\Psi$. Then w.p. $99\%$ over the choice of $\tilde{\Psi}$, $\ocspval_\Psi - \epsilon \leq \ocspval_{\tilde \Psi} \leq \ocspval_\Psi + \epsilon$. 
\end{lemma}

\begin{proof}
    First, fix an assignment $\vecsigma \in \sym_n$ to $\tilde{\Psi}$. Let $\cE_\vecsigma$ denote the event that $|\ocspval_{\tilde \Psi}(\vecsigma) - \ocspval_\Psi(\vecsigma)|\geq \epsilon$. We claim that \[ \Pr[\cE_\vecsigma] \leq 2\exp(-100 n\log n). \] Indeed, for $j \in [\tilde{m}]$, let $C_j$ denote the $j$-th constraint sampled when sampling $\tilde{\Psi}$, and let $X_j$ denote the indicator for the event that $C_j$ is satisfied by $\vecsigma$. So $\tilde{m} \cdot \ocspval_{\tilde \Psi}(\vecsigma) = \sum_{i=1}^{\tilde m} X_i$. Multiplying both sides of the desired inequality by $\tilde{m}$, $\cE_\vecsigma$ is the event that $|\sum_{i=1}^{\tilde m} X_i - \tilde{m} \cdot \ocspval_\Psi(\vecsigma)| \geq \epsilon \tilde{m}$. Since $C_j$ is a uniform constraint from $\Psi$, $\Exp[X_j] = \ocspval_\Psi$. Since the $C_j$'s are sampled independently, by the Chernoff bound (\ref{lemma:chernoff}) with $\eta = \epsilon$ and $p = \ocspval_\Psi$ we have
    \begin{multline*}
        \Pr\left[\left\lvert\sum_{i=1}^{\tilde m} X_i - \tilde{m} \cdot \ocspval_\Psi\right\rvert \geq \epsilon \tilde{m}\right] \leq 2\exp(-2\epsilon^2 \tilde{m} / \ocspval_\Psi) \\
        \leq 2\exp(-2\epsilon^2 \tilde{m}).
    \end{multline*}
    We assumed $\tilde{m} \geq 10 n\log n/\epsilon^2$. Therefore, the RHS is at most $2\exp(-100 n\log n)$, as desired.

    Now, let $\cE = \bigcup_{\vecsigma \in \sym_n} \cE_\vecsigma$ denote the event that $\cE_\vecsigma$ occurs for any $\vecsigma \in \sym_n$. To bound $\Pr[\cE]$, we use a union bound:  \[ \Pr[\cE] < 2 |\sym_n| \exp(-100n \log n). \] But by Stirling's approximation (\ref{lem:stirling}), $|\sym_n| \leq 3 \sqrt{n} (n/e)^n$, and $\exp(-100 n\log n) \geq n^{-100n}$ (since $\log n > \ln n$), so the error probability is $(3 \sqrt{n} (n/e)^n) (2 n^{-100n}) < 6 n^{-99n+\frac12}$ which is less than $1/100$ for sufficiently large $n$.

    Finally, we show that conditioned on $\cE$, the desired inequality $\ocspval_\Psi - \epsilon \leq \ocspval_{\tilde \Psi} \leq \ocspval_\Psi + \epsilon$ holds. Indeed, we have
    \begin{multline*}
        \ocspval_{\tilde \Psi} = \max_{\vecsigma \in \sym_n} \ocspval_{\tilde \Psi}(\vecsigma) \\ \leq  \max_{\vecsigma \in \sym_n} (\ocspval_\Psi(\vecsigma) + \epsilon) = \ocspval_\Psi + \epsilon
    \end{multline*} and similarly for $\ocspval_{\tilde \Psi} \geq \ocspval - \epsilon$.
\end{proof}

\begin{remark}
    For OCSPs, the $\log n$ factor in $\tilde{m}$ is required because the solution space has size $\Omega((n/e)^n)$ (by Stirling's approximation) and we take a union bound over all solutions. In contrast, for standard CSPs over an alphabet size of $q$, the solution space has size only $q^n$, and in the analogous analysis, there would only be a $\log q$ factor in $\tilde{m}$.
\end{remark}

Given \ref{lem:sparse}, we can now prove \ref{thm:sparse-alg}:

\begin{proof}[Proof of \ref{thm:sparse-alg}]
    Let $\epsilon' = \epsilon \rho(\Pi)/2$, and let $\tilde{m} = \lceil 10 n \log n/(\epsilon')^2 \rceil$. Consider the following streaming algorithm to sample an instance $\tilde{\Psi}$ on $n$ variables and $\tilde{m}$ constraints: Initialize a buffer $(C_1,\ldots,C_{\tilde{m}})$ of constraints, and then when the $i$-th constraint $C$ of $\Psi$ arrives, set $C_j \gets C$ with probability $1/i$ (otherwise $C_j$ remains the same) independently for each $j \in [\tilde m]$. After the stream, let $\tilde{\Psi}$ denote the instance formed by $(C_1,\ldots,C_{\tilde m})$, and output $\ocspval_{\tilde \Psi} - \epsilon'$.
    
    After all constraints of $\Psi$ arrive, each constraint $C_j$ is an independent, uniformly chosen constraint from $\Psi$, and so we can apply \ref{lem:sparse} to $\tilde{\Psi}$ to deduce that w.h.p. $\ocspval_\Psi - \epsilon' \leq \ocspval_{\tilde \Psi} \leq \ocspval_\Psi + \epsilon'$. Conditioning on this event, we deduce that \[ \ocspval_\Psi - 2\epsilon' \leq \ocspval_{\tilde \Psi}  \leq \ocspval_\Psi. \] The LHS is \[ \ocspval_\Psi - 2\epsilon' = \ocspval_\Psi - \epsilon \rho(\Pi) \geq (1-\epsilon) \ocspval_\Psi \] since $\ocspval_\Psi \geq \rho(\Pi)$, yielding the desired conclusion.
\end{proof}

\end{document}